\documentclass[10pt]{amsart} 

\usepackage{fullpage}

\usepackage{url}

\usepackage{amssymb}

\usepackage{cite}

\usepackage{graphicx}

\usepackage[usenames]{color}

\usepackage{accents}

\renewcommand{\a}{\alpha}
\renewcommand{\b}{\beta}
\newcommand{\g}{\gamma}
\renewcommand{\d}{\delta}
\newcommand{\D}{\Delta}

\newcommand{\f}{\varphi}
\newcommand{\s}{\sigma}
\newcommand{\Si}{\Sigma}

\renewcommand{\l}{\lambda}
 
\renewcommand{\t}{\theta}
\renewcommand{\O}{\Omega}
\renewcommand{\o}{\omega}

\newcommand{\cF}{{\mathcal F}}
\newcommand{\cC}{{\mathcal C}}

\newcommand{\cB}{{\mathcal B}}

\newcommand{\cV}{{\mathcal V}}

\newcommand{\cG}{{\mathcal G}}
\newcommand{\cH}{{\mathcal H}}

\newcommand{\cI}{\mathcal I}

\newcommand{\bR}{\mathbb R}

\newcommand{\bP}{\mathbb P}

\newcommand{\be}{\begin{equation}}
\newcommand{\ee}{\end{equation}}

\newcommand{\tr}{\mathrm{tr}}

\newcommand{\beaa}{\begin{eqnarray*}}
\newcommand{\bea}{\begin{eqnarray}}
\newcommand{\beal}[1]{\begin{eqnarray}\label{#1}}
\newcommand{\bean}{\begin{eqnarray}\nonumber}
\newcommand{\beadl}[1]{\begin{deqarr}\label{#1}}
\newcommand{\eeadl}[1]{\arrlabel{#1}\end{deqarr}}
\newcommand{\eeal}[1]{\label{#1}\end{eqnarray}}
\newcommand{\eead}[1]{\end{deqarr}}
\newcommand{\eea}{\end{eqnarray}}
\newcommand{\eeaa}{\end{eqnarray*}}

\renewcommand{\to}{\rightarrow}

\DeclareMathOperator{\Met}{Met}

\renewcommand{\div}{\mathrm{div}}

\renewcommand{\phi}{\varphi}
\renewcommand{\epsilon}{\varepsilon}
\renewcommand{\hat}{\widehat}

\newcommand{\<}{\langle}
\renewcommand{\>}{\rangle}
\newcommand{\p}{\partial}

\newcommand{\dm}{{\partial M}}

\newcommand{\w}{\widetilde}

\theoremstyle{plain}
\newtheorem{theorem}{Theorem}[section]

\newtheorem{remark}[theorem]{Remark}

\newtheorem{proposition}[theorem]{Proposition}

\newtheorem{conjecture}[theorem]{Conjecture}

\theoremstyle{definition}
\newtheorem{definition}[theorem]{Definition}

\def\blacksquare{\hbox to .60em {\vrule width .60em height .60em}}

\numberwithin{equation}{section}

\begin{document}

\title[ ]{The Initial boundary value problem and quasi-local Hamiltonians in general relativity}

\author{Zhongshan An and Michael T. Anderson}

\address{Dept. of Mathematics,
University of Connecticut, 
Storrs, CT 06269}
\email{zhongshan.an@uconn.edu} 

\address{Dept. of Mathematics, 
Stony Brook University,
Stony Brook, NY 11790}
\email{michael.anderson@stonybrook.edu}

\begin{abstract}
We discuss relations between the initial boundary value problem (IBVP) and quasi-local Hamiltonians in GR. The latter have 
traditionally been based on Dirichlet boundary conditions, which however are shown here to be ill-posed for the IBVP. We present 
and analyse several other choices of boundary conditions which are better behaved with respect to the IBVP and carry out a 
corresponding Hamiltonian analysis, using the framework of the covariant phase space method. 
\end{abstract}

\maketitle 

\section{Introduction} 
\setcounter{equation}{0}

  This article is concerned with the initial boundary value problem (IBVP) for the vacuum Einstein equations and its relation to 
the existence of quasi-local Hamiltonians in general relativity. We consider then space-times $M$ topologically 
of the form $M \cong I\times S$, where $I=(-1,1)$ and $S$ is a compact 3-manifold with non-empty boundary $\p S = \Si$. Let 
$\cC = I\times \Si$ denote the boundary $\p M$ of $M$ with $\Si = \p S$. The initial boundary value problem is the problem 
of finding Lorentz metrics $g$ on $M$ satisfying the vacuum Einstein equations 
\be \label{vac}
\mathrm{Ric}_g = 0,
\ee
together with prescribed boundary conditions along $\cC$ and initial conditions along a Cauchy surface (e.g.~$S = \{0\}\times S \subset 
M$). More precisely, one would like to establish existence, uniqueness and stability of solutions with prescribed initial and boundary 
data.  

  The analogous situation for the Cauchy or initial value problem has long been well-understood and it is worthwhile to recall this 
briefly to gain the general perspective. The initial data $(g_S, K)$ on $S$ consist of a Riemannian metric $g_S$ and symmetric 
bilinear form $K$ satisfying the vacuum Einstein constraint equations; these give the induced metric and second fundamental 
form of a solution $g$ on $S$. 

  Let $(V,g)$ denote a vacuum development of an initial data set $(S, g_S, K)$, i.e.~a globally hyperbolic vacuum 
spacetime $V$ containing $(S, g_S, K)$. It was proved by Choquet-Bruhat \cite{CB} that vacuum developments always exist. 
Two vacuum developments $(V_1, g_2)$, $(V_2, g_2)$ are called equivalent if they contain a common sub-development. 
Choquet-Bruhat and Geroch \cite{CG} further proved that there is a maximal development $V_{\max}$, unique up to isometry, 
i.e.~unique up to the action of the group $\mathrm{Diff}_0(V_{\max})$ of diffeomorphisms of $V_{\max}$ fixing $S$ pointwise. 
Let $\cV$ be the space of such isometry classes of maximal solutions. 

The main results on the solution of the Cauchy problem can then be summarized in the following statement: for each Cauchy surface 
$S$, there is a bijective correspondence  
\be \label{mod1}
E_S: \cV \to \cI, \ \ g \to (g_S, K),
\ee 
where $\cI$ is the space of all initial data $(g_S, K)$ satisfying the vacuum constraint equations. 
The bijection \eqref{mod1} gives an effective parametrization of the space of maximal solutions of \eqref{vac} by their initial data. 
It would appear to be likely that $\cI$ and hence $\cV$ may be given the structure of a smooth infinite dimensional manifold, possibly 
away from a small singular set of special data. Further by work of Fischer, Marsden and Moncrief, cf.~\cite{FM}, $\cV$ or $\cI$ should 
carry a naturally defined (non-degenerate) symplectic form $\O$.\footnote{Such manifold and symplectic structure results hold when 
$S$ is a closed or asymptotically flat 3-manifold, cf.~\cite{FM}, but this has not yet been extended in detail to the situation of manifolds 
with boundary. Also, for simplicity, we forgo here any detailed discussion of the specific function spaces, e.g.~Sobolev spaces, and 
related issues of regularity. The topologies and all geometric data are assumed to be $C^{\infty}$ in the following.}

\medskip 
  One would like to obtain similar results and a similar understanding for the IBVP, and so in particular obtain a bijective 
correspondence 
\be \label{mod2}
D_S: \cV \to \cI \times_{c}\cB, \ \ g \to ((g_S, K), b)
\ee 
where $\cV$ now denotes the space of maximal globally hyperbolic vacuum spacetimes $(M, g)$ with time-like boundary $\cC$ 
and $\cB$ is a space of boundary data on $\cC$. The subscript $c$ denotes compatibility or corner conditions arising in 
solutions of the IBVP.  Again, a maximal solution should be unique up to isometry, i.e.~up to the action of $\mathrm{Diff}_0(M)$ of 
diffeomorphisms of $M$ fixing $S$ and $\cC$ pointwise. Of course to even start, this requires finding a suitable set of 
boundary data $\cB$ on $\cC$. 

  The issue of the solvability of an IBVP (with certain boundary conditions) was first addressed in the foundational paper of 
Friedrich-Nagy \cite{FN}. This was followed by work of Kreiss-Reula-Sarbach-Winicour \cite{KRSW1}, \cite{KRSW2} and more 
recently by works of Fournadavlos-Smulevici \cite{FS1}, \cite{FS2}; we refer to \cite{ST} for a general survey of this topic. 
Unfortunately, none of these works establishes a well-posedness result as expressed in \eqref{mod2}. 

\medskip 

  As with the Cauchy problem, the most natural boundary conditions are geometric boundary conditions, i.e.~those determined by the 
$\mathrm{Diff}_0(M)$-invariant Cauchy data at $\cC$, i.e.~the induced metric $\g$ and second fundamental form $A$ of $\cC$ in $(M, g)$. 

  Based on the familiar situation with simpler field theories, e.g.~scalar or Yang-Mills gauge-type fields, it is usually assumed in the 
physics literature that the boundary conditions for vacuum gravity are of Dirichlet or Neumann type. For theories which have a fixed 
space-time structure (so the metric field is a background field, not subject to variation), Dirichlet or Neumann data are well known 
to be well-behaved and are well-posed for the IBVP, cf.~\cite{M}, \cite{SS}, \cite{W} for instance. In particular one generally has a 
correspondence as in \eqref{mod2} for such boundary data. 

  On the other hand, in general relativity there are no fixed background fields and we will show below that the 
straightforward analogy does not hold: Dirichlet boundary data (fixing the induced metric on the boundary $\cC$) or Neumann boundary 
data (fixing the second fundamental form $A$ of $\cC$ in $M$) are ill-behaved or ill-posed boundary data for vacuum gravity.\footnote{We 
note that the recent result of \cite{FS2} for totally geodesic boundary data $A = 0$ does not prove well-posedness for general Neumann 
boundary data, which is the situation considered here.} This is proved in Proposition 2.1, cf.~also Proposition 2.2, below. We conjecture 
that the conformal-mean curvature (CH) boundary data 
\be \label{come}
([\g], H),
\ee
consisting of the conformal class of the boundary metric $(\cC, \g)$ and its mean curvature $H = \tr A$, do lead to a well-posed IBVP 
for which \eqref{mod2} holds, cf.~Conjecture 2.3. Some evidence for this conjecture is the result, given in Proposition 2.4, that the 
vacuum constraint equations are naturally solvable under the CH boundary conditions \eqref{come}. 

  However, currently, the only known boundary data for which the IBVP is known to be well-posed and for which \eqref{mod2} holds are 
the diffeomorphism-invariant boundary data (called AA data here) recently developed in \cite{AA}; these are also discussed in more 
detail in \S 2. 

\medskip 

   Next, we relate the issues above with a Hamiltonian analysis of GR on finite manifolds-with-boundary. 
Hamiltonian analyses of GR on manifolds with boundary at infinity are well-known and well-developed. The ADM and 
Regge-Teitelboim Hamiltonian analysis for asymptotically flat space-times with boundary at spatial or null infinity are foundational 
works for the notions of global energy-momentum and angular momentum (the ``conserved charges") of such space-times.  
The case of time-like anti-de-Sitter boundaries at conformal infinity has also been extensively studied in connection with the 
AdS-CFT correspondence. In contrast, there has been much less detailed analysis in the case of finite, time-like boundaries, 
although see \cite{M2}, \cite{M1} for instance.  

  The main concern of this paper is then a more rigorous understanding of the definition of quasi-local energies (Hamiltonians) 
through an analysis of the behavior of the space of solutions of a well-posed IBVP. There is a large literature on Hamiltonian approaches to a 
suitable definition of quasi-local energy, starting with the foundational work of Brown-York \cite{BY}; we refer to \cite{Sz} for a 
detailed survey. These are all based on Hamiltonians with Dirichlet boundary conditions. The focus here is instead on both general 
boundary data, and specific boundary data, namely the CH data \eqref{come} and the AA boundary data developed in \cite{AA}.

The traditional canonical phase space approach toward a Hamiltonian analysis breaks the general covariance of the theory by a 
choice time evolution and corresponding time-like vector field $\p_t$. This gives a space-like foliation or 3+1 decomposition of the 
space-time, and an associated lapse-shift pair $(N, X)$, $\p_t = NT + X$. Hamiltonians $\cH$ thus depend on this data, 
$\cH = \cH_{(N,X)}$, cf.~\cite{Sz}, \cite{Wa}. 

  A very useful alternative to this approach which does not break the full covariance of the theory is the covariant phase space method, 
which directly gives a (pre)-symplectic structure to the space of all vacuum solutions (with given boundary conditions).  
We refer to the survey article \cite{K} for a detailed description of the background and history of this method; important contributions 
to this topic are those of  Crnkovic-Witten \cite{CW} and in particular Lee-Wald \cite{LW}, Iyer-Wald \cite{IW1}, \cite{IW2} and Wald-Zoupas 
\cite{WZ}. 

   Here we are interested in the context of space-times $M$ with time-like boundary $\cC$ which has been discussed much less in the 
literature. Fortunately, an excellent description of the covariant phase space method with boundary has recently been given by Harlow-Wu 
\cite{HW} and we will generally follow this description (pointing out however a number of subtleties). In \S 3, we provide an overview of 
this method, as applied to general (not only Dirichlet) boundary data. We also develop a modification of the method needed to deal with 
the diffeomorphism-invariant AA boundary data from \cite{AA}. 

  In Propositions 4.1 and 4.3 of \S 4, we show that the CH boundary data and a slight modification of the AA boundary data 
admit Lagrangian descriptions and have a well-defined variational formulation. In Propositions 4.4 and 4.5, we provide a detailed 
description of the corresponding phase space and derive the associated Hamiltonians. 

  In \S 5, we then discuss the basic issue of the choice of zero-point energies for Hamiltonians. For Dirichlet boundary data, 
these have been based on Euclidean (in the case of Brown-York \cite{BY}) or Minkowski (in the case of Wang-Yau \cite{WY}) subtraction 
terms. We consider the situation of general boundary conditions and a general definition of subtraction term in Definition 5.2, reminiscent 
of the Bartnik definition \cite{Ba} of quasl-local energy.

\section{Well-posed boundary data in GR}
\setcounter{equation}{0}

  The primary method to solve an IBVP for the Einstein equations is to find a gauge in which the reduced Einstein equations form 
a quasilinear hyperbolic system and then determine boundary conditions in which reduced system is a well-posed IBVP. Of course 
one also needs to ensure that all solutions of the reduced Einstein equations are true vacuum solutions; this is the issue of constraint 
preservation. For such a method of solution, a necessary condition for well-posedness is that the IBVP localized at a standard flat Minkowski 
corner region ${\bf R} = \{t \geq 0, x^1 \leq 0\} \subset \bR^{1,3}$ is well-posed, i.e.~the frozen coefficient IBVP is well-posed. All 
known existence and uniqueness results for the vacuum IBVP are based on this method. 

In this respect, we note the following: 

\begin{proposition}
With respect to either Dirichlet or Neumann boundary conditions, the linearization of the vacuum equations \eqref{vac} at a standard Minkowski 
background is not a well-posed IBVP, for any choice of gauge reduction. 
\end{proposition}

\begin{proof} The proof uses only the structure of the vacuum Einstein constraint equations (not the full Einstein equations) on the 
boundary $\cC$; 
\be \label{scal}
|A|^2 - H^2 - R_{\g} = E(\nu, \nu) = 0,
\ee
\be \label{div}
\div (A - H\g) = 2E(\nu, \cdot) = 0.
\ee
Here $E$ is the Einstein tensor of $g$, $R_{\g}$ is the scalar curvature of the boundary metric $(\cC, \g)$ and $\div$ is the divergence 
operator on symmetric $(2,0)$ tensors on $(\cC, \g)$; $\div h = \tr \nabla h$. As above, $A$ is the second fundamental form of 
$(\cC, \g) \subset (M, g)$ and the mean curvature $H$ is its trace. 

   Consider first Dirichlet boundary data, where the boundary metric $\g$ is prescribed. The linearization of the Hamiltonian or 
scalar constraint \eqref{scal} in the direction of the variation $h$ of $g$ is given by 
$$2\<A'_h, A\> - 2\<A^2, h^T\> - 2HH'_h - R'_{h^T} = 0,$$
where $h^T$ is the variation of the induced metric on the boundary. The linearization of these equations at the standard solid corner region 
$({\bf R}, g_0)$ of Minkowski space gives $A = H = 0$, 
so that 
$$R'_{h^T} = 0.$$
Thus, at the linearized level, only scalar flat deformations $h^T$ of the boundary metric $\g_0 = g_0|_{\cC}$, $\cC = \{x^1=0\}$, are 
possible. This is an infinite dimensional restriction; for example, most conformal deformations are excluded. It follows that the linearized 
constraint \eqref{scal} is not solvable for generic variations of boundary data $h^T$; the same applies then to the (more restrictive) 
full Einstein equations. 

 Next consider Neumann boundary data, where the second fundamental form $A$ is prescribed. The linearization of the 
momentum constraint \eqref{div} is given by 
\be \label{scal'}
\div_{\g}(A_h' - H_h'\g - H h^T) = -(\div)_h'(A - H\g).
\ee
At the standard solid region of Minkowski space, this gives 
\be \label{div'}
\div_{\g_0}(A'_h - H'_h\g_0) = 0.
\ee
where $\div_{\g_0}$ is the divergence operator with respect to the standard background Lorentz metric $\g_0$ on $\bR^{1,2}$. Thus the 
linearization of \eqref{div} is not solvable for generic $A'_h$; the form $A'_h$ has 6 degrees of freedom while \eqref{div'} gives 
3 scalar conditions on $A'_h$. This shows again that the linearized Einstein equations are not generically solvable at $g_0$ 
for prescribed variations $A'_h$. 

\end{proof}

From this, one expects that the space of vacuum solutions having given boundary metric (Dirichlet boundary data) or given 
second fundamental form (Neumann boundary data) is generically empty; the space of boundary data admitting solutions $g$ is of 
infinite codimension in the space $\cB$ of all boundary data. While Proposition 2.1 does not quite prove this, it does show that 
any proof of well-posedness cannot be based on solving a gauged IBVP for a quasi-linear hyperbolic system which localizes. 
More generally, one expects the behavior of the space of solutions with respect to either Dirichlet or Neumann boundary data 
to be highly unstable; given one solution with (say) Dirichlet boundary data $\g$ on $\cC$, an open set of nearby solutions have boundary 
data only in a subset of infinite codimension in the space $\cB = \Met(\cC)$ of boundary data. In particular, there is no correspondence as in 
\eqref{mod2}. 

   We also point out that there is a strong failure of uniqueness naturally associated with the failure of existence above. 

\begin{proposition}
In harmonic gauge, solutions to the linearized IBVP for the Einstein equations with Dirichlet boundary data have an infinite-dimensional 
kernel. 
\end{proposition}
\begin{proof}
Consider the initial boundary value problem for vacuum Einstein equations with Dirichlet boundary data:
\begin{equation*}
\begin{split}
& \mathrm{Ric}_g=0\ \ {\rm in}~M,\ \  g|_{S}=g_S,~\tfrac{1}{2}L_{T}g|_{S}=K\ \ {\rm on}~S, \ \ g_{\cC}=\g\ \ {\rm on}~\cC.
\end{split}
\end{equation*}
Let $x^{\a} = (t, x^1, x^2, x^3)$ be a local chart near the corner of $M$, mapping to the corner region ${\bf R} \subset \bR^{1,3}$. 
Assume $x^{\a}$ are harmonic, so that $V^{\a} = \Box_g x^{\a} = 0$. The kernel of the linearization of the 
system above at the Minkowski corner ${\bf R}$ is then given by:
\be \label{elin}
\begin{split}
& \mathrm{Ric}'_h+\d^*(V'_h)=0\ \ {\rm in}~\mathbf R\\
& h_{\a\b}=0,~\p_t h_{\a\b}=0\ \ {\rm on}~\{t=0,x^1\leq 0\}\\
&V'_h=0,~h^T=0\ \ {\rm on}~\{x^1=0\},
\end{split}
\ee
where $V'_h=-\div h + \frac{1}{2}d\tr h$. We show that the system \eqref{elin} has an infinite-dimensional space of solutions. 

  Set $h_{\a\b}=0$ in $\mathbf R$ for all $\a,\b$ except $(\a,\b)=(1,1),(0,1)$ and let 
\be \label{h01}
h_{01}=\tfrac{1}{2}h_{11}=f(x^1+t)
\ee
for some $f$ determined below. Then clearly $(\p_t-\p_{x^1})h_{11}=0,\p_{x^A} h_{11}=0$ for $A=2,3$. It follows that the bulk 
equation in \eqref{elin} holds: $(\mathrm{Ric}'_h+\d^*(V'_h))_{\a\b} = \Box h_{\a\b}=0$ in ${\bf R}$, for all $\a,\b$. 

Let $f: \bR \to \bR$ be any smooth function such that $f(x)=0, f'(x)=0$ for all $x\leq 0$ but $f(x)\neq 0$ for $x>0$. Then on the 
initial surface $\{t=0,x^1\leq 0\}$, $h_{01}=h_{11}=\p_th_{01} = \p_th_{11} = 0$, so that $h_{\a\b}$ satisfies the initial conditions. 
It remains to verify the boundary conditions. The Dirichlet boundary condition $h^T=0$, which reads $h_{00}=h_{0A}=h_{AB}=0$ 
for $A,B=2,3$, is obviously true. Further simple calculation also gives $V'_h = 0$, cf.~\cite[eqn. (3.6)]{AA} for example. (This is 
the reason for the particular relation between $h_{01}$ and $h_{11}$ in \eqref{h01}). 

Thus the deformation $h$ constructed above solves the homogeneous IBVP; however, $h_{11}\neq 0$ in the region $\{t+x^1>0\}\subset{\bf R}$. 
In other words, incoming wave solutions generated by such $f$ show that solutions to the linearized IBVP \eqref{elin} are not unique outside the 
domain of dependence of the initial surface. 
\end{proof}

Note that the same result and proof also work well for Neumann boundary data. 

\medskip 

  This raises the question of whether there are any geometric boundary data on $\cC$ for which the IBVP could be well-posed. We note that 
in the simpler Riemannian or Euclidean setting, it was proved in \cite{A1} that neither Dirichlet nor Neumann boundary data give a 
well-posed elliptic boundary value problem for the vacuum Einstein equations, in any gauge. On the other hand, it was shown there 
that the data 
\be \label{gH}
([\g], H),
\ee
consisting of the conformal class $[\g]$ of the boundary metric $\g$ and the mean curvature $H = tr A$, are well-posed 
elliptic boundary data, in a suitable gauge. (There exist other choices of well-posed elliptic boundary data, but the choice \eqref{gH} 
is perhaps the most natural). This was extended to the parabolic Ricci-flow setting in \cite{G}. 

  By analogy, it is natural then to consider the same data for the IBVP in the Lorentzian setting \eqref{vac}. We make 
the following 

\begin{conjecture}
The IBVP for the vacuum Einstein equations \eqref{vac} is well-posed with respect to the mixed boundary data $([\g], H)$ on $\cC$. 
\end{conjecture}
  
  There is strong evidence in favor of this conjecture, but a detailed discussion would be lengthy and is out of place here; we hope to 
address Conjecture 2.3 in detail elsewhere. 

  As a first justification for the conjecture, we note the following result, which contrasts strongly with Proposition 2.1. 
  
\begin{proposition}
The boundary data $([\g], H)$ are well-posed in terms of the constraint equations \eqref{scal}-\eqref{div}. In particular, the 
constraint equations on the boundary $\cC$ can be solved with any prescribed $([\g], H)$. 
\end{proposition} 

\begin{proof} This is proved by transferring the usual conformal method of Lichnerowicz-Choquet-Bruhat-York from the space-like 
elliptic setting to the time-like hyperbolic setting. 

Let $\g_0$ be a representative metric in $[\g]$ (not to be confused with the Minkowski metric on $\bR^{1,2}$ used above), with volume 
form $dv_{\g_0}$ and let $\a$ be a background volume form on $\cC$. Form then the so-called densitized lapse (cf.~\cite{Max}), 
$$\mu = {\tfrac{1}{2}}\frac{dv_{\g_0}}{\a}.$$
Let $\s$ be a transverse-traceless (2,0) tensor with respect to $\g_0$; $\div \s = \tr \s = 0$. 
In the conformal method, cf.~\cite{Max} again, one then constructs the Cauchy data $(\g, A)$ at $\cC$ by setting 
$$\g = \l^{4}\g_{0},$$
$$A = \l^{-2}(\s + {\tfrac{1}{2\mu}}\hat L_{X}\g_{0}) + {\tfrac{1}{3}}H\l^{4}\g_{0},$$
where $\hat L_{X}\g_{0}$ is the conformal Killing operator with respect to $\g_{0}$: $\hat L_{X}\g_{0} = L_{X}\g_{0} - 
(\frac{2}{3}\div_{\g_{0}}X) \g_{0}$. The constraint equations \eqref{scal}-\eqref{div} in this time-like setting then become 
a coupled system of equations for $(\l, X)$ which take the form 
\be \label{confm}
\begin{array}{c}
\div_{\g_0}({\tfrac{1}{2\mu}}\hat L_{X}\g_{0}) = {\tfrac{2}{3}}\l^{6}dH,\\
8\Box_{\g_0} \l = R_{\g_0}\l - |\s + {\tfrac{1}{2\mu}}\hat L_{X}\g_0|^{2}\l^{-7} + {\tfrac{2}{3}}H^{2}\l^{5}.
\end{array}
\ee
Here $\div_{\g_0}$ and $\Box_{\g_0}$ are the divergence and wave operator with respect to $\g_{0} \in [\g]$ and $R_{0}$ is the scalar 
curvature of $\g_{0}$. One has $\div \hat L_X(\g_0) = \div \d^*X - \frac{2}{3} d\div X$ and standard Weitzenbock formulas show that 
$\div \hat L_X(\g_0) = -D^*D X + \frac{1}{2}d\div X + E(X)$, where $E(X)$ involves only the curvature of $\g_0$ and no derivatives of $X$.   

   The system \eqref{confm} is thus a coupled system of semi-linear wave equations for $(\l, X)$ and so by standard theory is 
well-posed; one has existence and uniqueness of  solutions for at least some positive time $t \in [0, T)$, given suitable initial conditions. 

  It is also easy to see that the York decomposition \cite{Yo} holds on $(\cC, [\g])$; any symmetric bilinear form $h$, (e.g.~$h = A$ above) 
admits a decomposition $h = \s + \frac{\tr h}{3}\g_0 + \hat L_{X}\g_0$, where $\s$ is transverse-traceless with respect to $\g_0$. 
The vector field $X$ is unique up to choice of initial conditions on a slice $\Si$ to $\cC$. The proof of this is essentially the same as 
in the Riemannian or space-like case, using hyperbolic equations in place of elliptic equations.

\end{proof}

  Thus the space of solutions of the Einstein constraint equations on a time-like boundary $\cC$ is parametrized by the free 
data $([\g], H, \s)$ and initial conditions at $\Si$. This time-like boundary setting is actually significantly better behaved than the analogous 
situation for initial data, where the corresponding result is false. In general, given $([g_S], k, \s)$, $k = \tr_S K$, on an initial slice 
$S$, it is not always possible obtain existence of (unique) solutions to the constraint equations. The conformal method does not 
always work well for example when $H$ is far away from the set of constant functions (the far from CMC regime). 

\begin{remark}
{\rm  We make some further remarks comparing Propositions 2.1 and 2.4, focusing on the constraint equations with Dirichlet data. 

   Consider first initial data, i.e.~the issue of solving the constraint equations on a Cauchy surface with given Dirichlet data, i.e.~initial 
metric $\g$. In bulk dimension 3, this is the problem of isometric immersion of a given Riemannian surface $(S, g_S)$ into 
$\bR^{1,2}$, by the fundamental theorem of hypersurfaces in $\bR^{1,2}$, (at least when $S$ is simply connected). We are not aware 
of significant results on this topic. The analogous isometric immersion problem of a Riemannian surface in Euclidean space $\bR^3$ 
is notoriously hard and largely unresolved except in the case of positive Gauss curvature $K_{\g} > 0$; this is the solution of the 
Weyl embedding problem by Nirenberg and Pogorelov.\footnote{Note that in this Riemannian situation, the condition $K_{\g} > 0$ 
is an open condition, so that there are open sets of $C^{\infty}$ solutions of the constraint equations with given Dirichlet data, in 
strong contrast with Proposition 2.1. Nevertheless, the constraint equations \eqref{scal}-\eqref{div} are not a well-posed or elliptic 
system, even when $K_{\g} > 0$. Instead, the main idea is to quotient the space $\Met(S)$ by the positive definite second fundamental 
form $K$, where the problem does become (underdetermined) elliptic, cf.~\cite{Ham}, \cite{A1}, \cite{W}. This method is thus restricted 
to the convex case.}

   Passing to Dirichlet data on a Cauchy surface in 4 dimensions is the problem of isometrically immersing a given Riemannian 3-manifold 
into some vacuum space-time. Again, we are not aware of significant results on this problem.  

   It seems even less is known in the setting of time-like boundaries, i.e.~results or a theory of isometrically immersing Lorentzian 
surfaces into flat $\bR^{1,2}$ or similarly, a theory of isometrically immersing Lorentzian 3-manifolds into a vacuum spacetime. 
On the other hand, Proposition 2.4 shows the situation is very well understood for $([\g], H)$ data (for all dimensions). 
}
\end{remark}

\begin{remark}
{\rm  Once one solves the time-like constraint equations, if the data are analytic, then the Cauchy-Kovalevsky theorem can be 
used to obtain actual vacuum solutions defined in a neighbhorhood of $\cC$; this gives local two-sided solutions, i.e.~solutions 
extending to each side of the boundary $\cC$. In more detail, given a pair $(\g, A)$ of analytic data satisfying the constraint equations 
\eqref{scal}-\eqref{div}, one may work in the gauge of geodesic normal coordinates normal to $\cC$ where $g(\nu, \cdot) = \d_{\nu \, \cdot}$. 
In this gauge, the Einstein equations are ``evolution equations" off the boundary in the normal $\nu$-direction, and for analytic data the 
Cauchy-Kovalevsky theorem gives the existence of analytic solutions in a neighborhood of $\cC$. Of course the simple scalar 
model here is the Laplace equation $\D u = 0$ or wave equation $\Box u = 0$ with prescribed Cauchy data on $\cC$.  

 While this is not a well-posed boundary value problem, it is a delicate issue to understand whether a solution $g$ defined in a 
neighborhood of $\cC$ extends to a smooth solution on $M$. 
}
\end{remark}

   In the discussion of geometric boundary conditions above, note that the number of degrees of freedom for the boundary data 
($\g$, or $A$, or $([\g], H)$) is $6$. This is the correct and expected number for the gauge group $\mathrm{Diff}_0(M)$ of diffeomorphisms 
fixing $\cC$ pointwise, i.e.~(maximal) solutions with the same initial and boundary data are expected to be unique up to diffeomorphisms 
fixing an initial slice $S$ and the boundary $\cC$. However, the space of all solutions is invariant under the larger gauge group 
$\mathrm{Diff}(M)$, which essentially differs from $\mathrm{Diff}_0(M)$ by the group $\mathrm{Diff}(\cC)$ of boundary diffeomorphisms. 
The group $\mathrm{Diff}(\cC)$ has 3 degrees of freedom and so the number of true boundary conditions describing non-isometric 
solutions should be $6-3=3$. Of these $3$ degrees of freedom, $2$ account for the degrees of freedom of the gravitational field, 
and $1$ for the location or evolution of the boundary $\cC$ off the initial corner surface $\Si$. Note that it may not be easy to find a 
global slice to (for instance) the space of Dirichlet data $\Met(\cC)/\mathrm{Diff}(\cC)$.

\medskip 

  Next we discuss the boundary data introduced in \cite{AA}. This is based on the construction of a global preferred 
harmonic-type gauge for a metric $(M, g)$, at least near the boundary $\cC$. The preferred gauge depends on the 
space-time $(M,g)$ as well as a vector field $\Theta_{\cC}$ defined on $\cC$; different choices of $\Theta_{\cC}$ give 
different choices of time evolution or space-like foliation of the space-time $(M, g)$. The gauge is given by a wave map 
\be \label{wm}
\begin{array}{c}
\f_g: (M, g) \to (M_0, g_R),\\
\Box_g \f_g  = 0,
\end{array}
\ee
where for simplicity the target space $M_0$ is taken to be a standard solid cylinder $I \times B^3$ in Minkowski space $\bR^{1,3}$ with 
standard time function $t$. The Riemannian metric $g_R$ is the associated flat Euclidean metric. Initial data for $\f_g$ on a 
Cauchy slice $S$ are given by 
\be \label{id}
\f_g = E_{g_S}, \ \ (\f_g)_*(T_g) = T_{g_R} \ \ {\rm on} \ \ S,
\ee
where $T_g$ is the $g$-future unit normal to $S$ in $(M, g)$, $T_{g_R}$ is the future unit normal to $\{t = 0\}$ in $M_0$ and 
$E_{g_S}$ is determined by $g_S = g|_S$ equivariantly with respect to $\mathrm{Diff}(S)$ near the boundary $\Si$, i.e.
$$E_{\psi^*(g_S)} = \psi^*E_{g_S} = E_{g_S}\circ \psi.$$
We refer to \cite{AA} for details regarding the assignment $g_S \to E_{g_S}$ and for concrete examples of such. 

   The boundary data for $\f_g$ are of Sommerfeld type given by 
\be \label{bd}
[(\f_g)_*(T_g + \nu_g)]^T = \Theta_{\cC},
\ee 
together with the (scalar) Dirichlet condition that $\f_g: \cC \to \p M_0 = I\times S^2$. Here $\nu$ is the $g$-unit normal of 
$\cC \subset M$ and the superscript $T$ denotes orthogonal projection with respect to $g_R$ onto $\p M_0$. Thus, we view $\Theta_{\cC}$ 
as a vector field tangent to $\cC_0 = \p M_0$. 

The equations \eqref{wm} are hyperbolic and initial boundary value problems of this type are well-understood; it follows from such 
known results that the IBVP \eqref{wm}-\eqref{bd} is well-posed. In particular, given a background metric $g$, there exists a unique 
solution $\f_g$ which is a diffeomorphism in a neighborhood of the initial and time-like boundary $T = S \cup \cC$. Moreover (at least 
near $T$), the solution is equivariant with respect to diffeomorphisms, so that 
\be \label{equi1}
\f_{\psi^*g} = \psi^* \f_g,
\ee
for any $\psi \in \mathrm{Diff}(M)$. 

  Next consider the coupled system $(g, \f_g)$ 
\be \label{couple}
\begin{array}{c}
\mathrm{Ric}_g = 0, \\
\Box_g \f_g = 0,
\end{array}
\ee
with initial conditions for $g$ given by 
\be \label{idg}
g|_S = g_S, \ \ K_g |_S = K,
\ee
where $(g_S, K)$ satisfy the constraint equations \eqref{scal}-\eqref{div}. The boundary conditions for $g$ are expressed on the target 
$\cC_0 = \p M_0$ as: 
\be \label{bcg}
([\g^t], H^*, \Theta_{\cC}).
\ee
Here $\g^t = (\f_g^{-1})^*g|_{\Si_t}$, is the metric induced by the pullback $(\f_g^{-1})^*$ of $g$ to $M_0$ restricted to the 
slices $\Si_t = \{t = constant\}$ in $\cC_0$ and $[ \cdot ]$ denotes the conformal class of the 2-metric. The term $H^*$ is a combination 
of mean curvatures of the form 
\be \label{Hstar}
H^* = \a \tr_{\cC}A |_{\Si_t} + \b \tr_{\Si_t}A + \g \tr_{\Si_t}K,
\ee
for certain values of $\a, \b, \g$, (cf.~\cite{AA} for the exact allowable coefficients). One admissible choice of $H^*$, which also appears 
in \cite{KW}, is $H^* = 2\tr_{\cC}A - \tr_{\Si_t}A$; this choice will be discussed further in \S 4.  The initial conditions for $\f_g$ are 
as in \eqref{id}.

   It is proved in \cite{AA} that this is a well-posed IBVP. For any choice of compatible initial data \eqref{idg} and boundary data 
\eqref{bcg}, there is a vacuum solution $g$ realizing the initial and boundary conditions. Such solutions have the covariance property
\be \label{equi2}
\bar g = \psi_g^*(g) =\psi_{\chi^*g}^*  (\chi^*g),
\ee
for any $\chi \in \mathrm{Diff}(M)$, where $\psi_g = \f_g^{-1}$. This gives a slice (at least near the boundary $\cC$) to the action of 
$\mathrm{Diff}(M)$ on the space of vacuum solutions. Thus, diffeomorphisms no longer act on the metrics $\bar g$; these metrics are 
evaluated in the fixed standard coordinates $(t, x^i)$ on $\bR^{1,3}$. The solutions $\bar g$ are unique with given initial and boundary 
data and one has stability of solutions and the existence of maximal solutions. Moreover, any vacuum solution is isometric 
to such a solution in a preferred gauge. Note that since the boundary data $([\g^t], H^*, \Theta_{\cC})$ always 
admit solutions, it follows that the constraint equations \eqref{scal}-\eqref{div} on $\cC$ are automatically satisfied; in fact the 
form of the $H^*$ term in \eqref{Hstar} derives partly from the form of the constraint equations. 

    An important point is that, by construction, the boundary conditions \eqref{bcg} are invariant under the action of 
$\mathrm{Diff}(M)|_{\cC} = \mathrm{Diff}(\cC)$ on the space of vacuum solutions. However, the choice 
of $\Theta_{\cC}$ determines a particular or preferred gauge $\f_g$, giving a particular foliation or 3+1 decomposition of 
the space-time $(M, g)$ in which the boundary data are measured. The relevant space of (geometric) boundary data is then 
\be \label{bgeom}
\cB_{geom} = I\times \mathrm{Conf}(\Si)\times C^{\infty}(\Si) = \{ (t, [\g^t]_{\Si_t}, H^*_{\Si_t})\},
\ee
where $\mathrm{Conf}(\Si)$ is the space of pointwise conformal classes $[\g^t]$ of smooth metrics on $\Si$ and $H^* \in C^{\infty}(\Si)$. 
The space $\cB_{geom}$ has $3$ degrees of freedom, corresponding to $\mathrm{Diff}(\cC)$ invariant boundary data. 
This is in marked contrast with the discussion of geometric boundary data (Dirichlet, Neumann or $([\g], H)$) above, which are 
not $\mathrm{Diff}(\cC)$ invariant.

The well-posedness of this IBVP then gives a family of effective parametrizations as in \eqref{mod2}, i.e. 
\be \label{mod3}
D_{S, \Theta_{\cC}}: \cV_{geom} \to \cI_{geom} \times_c \cB_{geom},
\ee
depending on the choice of Cauchy surface and boundary gauge $\Theta_{\cC}$. The space $\cV_{geom}$ is now the space of full isometry 
classes of solutions (vacuum Einstein metrics), i.e.~the space of all solutions modulo the action of the full diffeomorphism group
 $\mathrm{Diff}(M)$. Accordingly $\cI_{geom}$ is the space of equivalence classes of Cauchy data $(g_S, K)$ on a Cauchy slice, 
 where $(g_S, K) \sim (g_S', K')$ if there is a space-time diffeomorphism $\chi \in \mathrm{Diff}(M)$, $\chi(S) = S$, such that 
 $(\chi^*g_S', \chi^*K') = (g_S, K)$. 

  A variational description and Hamiltonian analysis of this boundary data will be discussed in the following sections.

\section{The IBVP and Quasi-local Hamiltonians} 
\setcounter{equation}{0}

  In this section, we discuss the covariant phase space method applied to vacuum GR with a time-like boundary, following 
\cite{HW}, \cite{LW}, \cite{IW1}. 

Consider Lagrangians with boundary terms describing vacuum gravity, 
\be \label{I}
I(g) = I_{EH}(g) + \int_{\cC}\ell(g) = \int_M R_g dv_g + \int_{\cC}\ell(g).
\ee
(We have chosen units where $16\pi G = 1$). 
The bulk Einstein-Hilbert action $I_{EH}$ is invariant under the diffeomorphism group $\mathrm{Diff}(M)$; only boundary 
actions $\int_{\cC}\ell$ which are also $\mathrm{Diff}(\cC)$ invariant will be considered.  Different choices of boundary action 
correspond to different choices of boundary data space $\cB$. As a simple example, consider the case of a scalar field 
$u: M \to \bR$ with action given by 
\be \label{ID}
I_D(u) = \int_M ({\tfrac{1}{2}}|du|^2 + V(u))dv_g,
\ee
where $g$ is a fixed background globally hyperbolic Lorentz metric on $M$ and $V(u)$ is a function involving only 0-jet of $u$. 
Then the variation of $I_D$ in the direction $v = u'$ is given by
$$\frac{d}{dt}I_D(u+tv)|_{t=0} = -\int_M v (\Box u - V'(u))  - \int_{\cC}v \p_{\nu}u - \int_{S^+ \cup S^-}\Psi_u(v),$$
where $\Psi_u(v) = v\p_Tu$ and $T$ is the unit outward time-like normal to the future and past the slices $S^{\pm}$. This term is 
ignored in the variational problem, cf.~\cite{HW}, so that the variational derivative is given by 
$$\d I_D(v) = -\int_M v (\Box u - V'(u))  - \int_{\cC}v \p_{\nu}u.$$
The boundary condition that $I_D$ is stationary $\d I_D = 0$ then requires that $v = 0$ at $\cC$, i.e.~the Dirichlet data $u|_{\cC}$ of $u$ 
is fixed. Hence the space $\cB$ is the space of Dirichlet boundary data of $u$, $\cB = C^{\infty}(\cC)$. The corresponding equations of 
motion are  
\be \label{el}
\Box u + V'(u) = 0.
\ee
If one adds a boundary action of the form 
\be \label{IN}
I_N(u) = \int_M ({\tfrac{1}{2}}|du|^2 + V(u))dv_g + \int_{\cC}u\p_{\nu}u,
\ee
then a similar analysis shows that the boundary condition becomes $\p_{\nu}v = 0$, i.e.~the Neumann boundary data $\p_{\nu}u$ 
of $u$ is fixed, with corresponding boundary data space again $\cB = C^{\infty}(\cC)$.

  In either case above, one has a well-defined variational problem on the field or configuration space $\cF = C_b^{\infty}(M)$ of functions 
$u: M \to \bR$ with $u$ satisfying a fixed boundary condition $b \in \cB$ (Dirichlet or Neumann) at the boundary $\cC$; namely, 
$I_D$, resp. $I_N$, is a smooth function on $C^{\infty}(M)$ whose critical points are solutions of the Euler-Lagrange equations 
\eqref{el}. Moreover, standard theory of the IBVP for the (non-linear) wave equation \eqref{el} implies that the space of solutions $\bP$ of 
\eqref{el} satisfies the analog of \eqref{mod2}, 
$$\bP \simeq \cI \times_c \cB.$$
The covariant phase space method discussed below defines a natural symplectic structure on $\bP$, with corresponding Hamiltonians, 
cf.\cite{CW}, \cite{HW} for further details. 

\medskip 
We emphasize here that the issues of having a well-defined variational problem for a field theory and a well-posed IBVP for 
the equations of motion are fundamentally distinct and independent. In particular, the property of having a well-defined variational 
problem does not at all imply the property of having a well-posed IBVP. 
 
\medskip 

  Returning to GR and the action \eqref{I}, let $\cF = \Met(M)$ be the configuration space of all maximal globally hyperbolic Lorentz metrics 
on $M$ with time-like boundary $\cC$. The pre-phase space $$\w \bP \subset \cF$$
is defined to be the space of globally hyperbolic solutions of the vacuum Einstein equations on $M$. 
(The terminology on-shell for data in $\w \bP$ and off-shell for data in $\cF$ is often used). The bulk diffeomorphism 
group $\mathrm{Diff}(M)$, mapping $\cC \to \cC$, acts naturally on both spaces by pullback. Although it is geometrically 
natural to consider the quotient space of isometric metrics $\Met(M)/\mathrm{Diff}(M)$, it will be seen below that the quotient 
$\w \bP / \mathrm{Diff}(M)$ does not generally serve as a suitable candidate for the phase space. 

  We first discuss the ``bare" situation without any boundary terms. For the bare EH action, 
$$I_{EH}(g) = \int_{M}R_g dv_g,$$
one has $\d_g (R_gdv_g)(h) = -\<E_g, h\>dv_g  + d\t(h)$ where $d\t(h) =( - \Box tr h + \div \div h)dv_g$. The 3-form 
\be \label{theta}
\t(h) = -\star(d\tr h - \div h),
\ee
is called the pre-symplectic potential and plays a central role in this discussion. Note that for any 2-form $C$ on $M$, 
$d(\t - dC) = d\t$.  On-shell, i.e.~tangent to $\w \bP$, $(R_gdv_g)' = 0$ and $E_g = 0$, so that, on $T\w \bP$, 
\be \label{tc}
d \t(h) = 0
\ee
pointwise on $M$. In other words, $\t$ is a closed 3-form on $M$ on-shell. It follows that 
\be \label{v1}
\frac{d}{dt} I_{EH}(g+th)|_{t=0} = 0,
\ee
on-shell. Thus (in contrast to the situation in \eqref{ID} or \eqref{IN}) it might appear that there are {\it no} boundary conditions 
imposed by the requirement of stationarity. However, \eqref{v1} only gives $\int_{\dm}(\t - dC)(h) = 0$, i.e. 
$$\int_{\cC} (\t - dC)(h) = \int_{S^+}(\t - dC)(h) - \int_{S^-}(\t - dC)(h).$$
This shows there is a (global) coupling of the variation at the future and past Cauchy surfaces with the variation at the boundary $\cC$. 
For a well-defined variational problem, one requires instead that 
\be \label{v0}
\d I_{EH}(h) = \int_{\cC}(\t - dC)(h)= 0,
\ee
for all on-shell variations $h$. 

   On $\cC$, $\t(h) = -(\nu(\tr h) -\div h(\nu))dv_{\g}$ is a gauge, i.e.~$\mathrm{Diff}_0(M)$ dependent term. To obtain a 
geometric or gauge-invariant boundary condition, recall the basic identity (valid on any hypersurface) 
\be \label{h}
-\t(h) =\big(\nu(\tr h)  - (\div h)(\nu) \big)dv_{\g}= \big(2H'_h + \<A, h\> + \div(h(\nu)^T)\big)dv_{\g}.
\ee
It is natural to set 
\be \label{Ch}
C(h) = \star h(\nu)^T \ \ {\rm on} \ \  \cC, 
\ee 
and we do so for the remainder of the paper; the discussion to follow is independent of the extension of $C$ off $\cC$. Also set 
\be \label{hatt}
\hat \theta = \theta - dC.
\ee
Then by \eqref{h}, \eqref{v0} is equivalent to 
\be \label{v2}
\int_{\cC}(\<A, h\> + 2H'_h)dv_{\g} = 0,
\ee
This corresponds to geometric boundary conditions $h^T = H'_h = 0$. However, these are now $7$ boundary conditions and so these 
are over-determined boundary conditions, (cf.~the discussion in \S 2). 

   This shows that it is necessary to include boundary Lagrangians $\ell$ to obtain a well-defined variational problem. To begin, the 
action \eqref{I} is stationary on the subspace of $h \in T\w \bP$ such that 
\be \label{bc}
[(\t - dC +\d \ell)|_{\cC}](h) = 0.
\ee
Thus associate to the action \eqref{I} a space of (geometric) boundary data $\cB$ and consider the natural evaluation or restriction map 
$$R: \w \bP \to \cB,$$
equivariant with respect to the action of diffeomorphisms. Given a choice of boundary data $b \in \cB$, the pre-phase space 
$\w \bP_b$ is defined as the inverse image 
\be \label{pre}
\w \bP_b = R^{-1}(b):
\ee
$\w \bP_b$ is the space of vacuum metrics with boundary data in $b$. Then \eqref{bc} is equivalent to  
$$T \w \bP_b = {\rm Ker}\, ([\t - dC +\d \ell]_{\cC}).$$
Via \eqref{h}, the boundary condition \eqref{bc} is thus equivalent to 
\be \label{bc3}
-\<A, h\> - 2H'_h + \ell'_h = 0 \ \ {\rm on} \ \ \cC.
\ee

Of course in general $\w \bP_b$ will no longer have an effective induced action of $\mathrm{Diff}(M)$; the choice of $b$ 
breaks the symmetry, i.e.~breaks the isometric action of $\mathrm{Diff}(M)$ on $\Met(M)$ to that of the subgroup $\mathrm{Diff}_b(M)$ 
preserving the boundary condition $b \in \cB$.

  The pre-symplectic current $\o$ on $\w \bP_b$, or more generally on $\w \bP$ or $\cF$, is given by 
\be \label{omega}
\o = \d \hat \t ,
\ee
where $\d$ is the exterior derivative on $\cF$. Thus,  
$$\o(h_1, h_2) = (\hat \t(h_2))'_{h_1} - (\hat \t(h_1))'_{h_2},$$
where $h_1$, $h_2$ are tangents to a 2-parameter deformation in $\w \bP_b$, (so that the formal correction term 
$\t([h_1, h_2])$ vanishes). 

Given a Cauchy surface $S$, the pre-symplectic form is a 2-form on $\w \bP_b$ defined by 
$$\w \O(h_1, h_2) = \int_S \o(h_1, h_2) = \int_S (\hat \t(h_2))'_{h_1} - (\hat \t(h_1))'_{h_2}.$$
Using \eqref{h}, this gives 
$$\w \O(h_1, h_2) = -\int_S \<K'_{h_1}, h_2\> - \<K'_{h_2}, h_1\> + {\tfrac{1}{2}}[tr_S h_1(\<K, h_2\> + 2k'_{h_2})  - 
tr_S h_2(\<K, h_1\> + 2k'_{h_1})],$$ 
where $K$ is the extrinsic curvature of $S$ in $(M,g)$ and $k = tr_S K$. Writing $\pi = (K - k g_S)dv_{g_S}$, where $K - k g_S$ is 
the momentum conjugate to $g_S$, this may be written (in standard form) as 
\be \label{symp2}
\w \O(h_1, h_2) = -\int_S \<\pi'_{h_1}, h_2\> - \<\pi'_{h_2}, h_1\>.
\ee
This gives a non-degenerate symplectic form off-shell, i.e.~on the configuration space $\cF$, cf.~\cite{LW}. However, the vacuum constraints 
and boundary conditions imply that $\w \O$ is degenerate on-shell, i.e.~on $\w \bP$ or $\w \bP_b$. 

 The boundary condition \eqref{bc} implies $\o = \d \d \ell = 0$ on $\cC$. Since $\o$ is a closed $3$-form on $M$, it follows 
that $\w \O$ is well-defined, independent of the Cauchy surface $S$, on $\w \bP_b$. This is not the case without the boundary 
condition \eqref{bc}; thus $\w \O = \w \O_S$ gives a pre-symplectic form on the full space $\w \bP$ which depends on the 
Cauchy surface $S$. 

 Hamiltonians are functions $\cH: \w \bP_b \to \bR$ for which there is a vector field $X_{\cH}$ such that 
\be \label{Ham}
\w \Omega(X_{\cH}, \cdot) = \d \cH.
\ee
If $\w \O$ were non-degenerate, then any smooth function $\cH: \w \bP_b \to \bR$ is a Hamiltonian. Namely, setting $\w \O = \O$ in 
the non-degenerate situation, the symplectic gradient $X_{\cH}$ is given by 
$$X_{\cH} = \O^{-1}(\d \cH, \cdot),$$
where $\O$ defines the pairing: $T\bP \to T^*\bP$. If, as it is here, $\w \O$ is degenerate, then Hamiltonians may not exist. 
A Hamiltonian $\cH$ must be invariant under the flow generated by the degenerate directions ${\rm Ker} \,\w \O$. Moreover, $\cH$ 
is always conserved, i.e.~$\cH$ is invariant under the flow of the Hamiltonian vector field $X_{\cH}$ on $\w \bP_b$. 

\medskip 

The actual phase $\bP_b$ with boundary data $b \in \cB$ is given by a standard symplectic reduction process, cf.~\cite{MW}, as the 
formal quotient  
\be \label{qg}
\bP_b = \w \bP_b/\cG,
\ee
where $\cG$ is the Lie group generated by the Lie algebra ${\rm Ker}\, \w \O$. Then $\w \O$ descends to a (non-degenerate) symplectic 
form on $\bP_b$, giving it formally a symplectic structure. 

As the primary example of interest here, consider a vector field $\xi$ generating a flow in $\mathrm{Diff}(M)$ by spacetime 
diffeomorphisms of $M$. This generates a flow on $\w \bP_b$ in the usual way by pullback, with tangent vector $L_{\xi}g \in T_g\w \bP_b$.  

As shown in \cite{LW}, \cite{IW1}, \cite{HW}, the Hamiltonian 
$\cH_{\xi}$ is given by
\be \label{Ham1}
\cH_{\xi} = \int_{\Si} q_{\xi} + \xi \cdot \ell - C(L_{\xi}g),
\ee
where $q_{\xi}$ is the Noether potential given by 
$$q_{\xi} = -\star d\xi,$$
viewing $\xi$ as a 1-form on $\cC$. We consider $\ell$ as a scalar density, $\ell = \ell dv_{\g}$ on $\cC$, so that $\xi \cdot \ell = 
\xi\lfloor dv_{\cC}$. On $\Si$, $\star d \xi = d\xi(T, \nu) = \<\nabla_{T}\xi, \nu\> - \<\nabla_{\nu}\xi, T\>$, 
where $\nu$ is as above and $T$ is the unit normal to $\Si \subset \cC$. Since $L_{\xi}g(T, \nu) = \<\nabla_{T}\xi, \nu\> + 
\<\nabla_{\nu}\xi, T\>$ and $\<\nabla_{T} \xi, \nu\> = -A(\xi, T)$, this gives 
\be \label{Ham2}
\cH_{\xi} = \int_{\Si}(-2A(\xi, T) + \ell\<\xi, T\>)dv_{\g_{\Si}}. 
\ee
Note that $\cH_{\xi}$ is independent of the slice $\Si$ since the variation $h = L_{\xi}g$ satisfies the boundary 
condition \eqref{bc}, i.e.~$L_{\xi}g \in T \w \bP_b$. Otherwise, as with $\w \O$, this will not be the case for general 
$\xi \in T\mathrm{Diff}(M)$. 

  Clearly if $\xi = 0$ on $\cC$ (or on suitable domains in $\cC$), then $\cH_{\xi} = \d \cH_{\xi} = 0$. so the symplectic form $\w \O$ is 
degenerate on vectors tangent to $\mathrm{Diff}_0(M)$. Thus one always has 
\be \label{og}
\mathrm{Diff}_0(M) \subset \cG.
\ee
On the other hand, the existence of non-trivial Hamiltonians generated by $\xi$ as above shows that generally $\mathrm{Diff}_b(M)$ is not 
contained in $\cG$. 
 
  Finally we note that a Hamiltonian $H_{\xi}: \bP_b \to \bR$ as in \eqref{Ham1} is determined by \eqref{Ham} only up to an 
additive constant; the constant may depend on the boundary condition $b \in \cB$. The choice of the normalizing 
constant plays an important role in Hamitonian approaches to the issue of quasi-local energy and is discussed in more detail in \S 5.

\medskip 

  The covariant phase space method is related to the perhaps more well-known canonical phase space method by the choice of a 
Cauchy surface $S$. Namely, given such an $S$, one has a {\it formal} identification 
\be \label{ES}
E_S: \hat \bP_b \simeq \cI, \ \ E_S(g) = (g_S, K),
\ee
where $\hat \bP_b = \w \bP_b / \mathrm{Diff}_0(M)$ and $\cI$ is the space of initial data $(g_S, K)$ satisfying the constraint 
equations. Here $\mathrm{Diff}_0(M)$ is the group of diffeomorphisms of $M$ fixing $T = S \cup \cC$. Of course different choices 
of Cauchy surface give different parametrizations $E_S$. The data $(g_S, K)$, or more precisely $(g_S, \pi)$, serve as the ``$q$, $p$ 
variables" for equivalence classes of solutions in $\hat \bP_b$ and so such data effectively serve as a choice of coordinates for 
$\hat \bP_b$. Solutions $g \in \w \bP_b$ are then given by a specific choice of time evolution along a curve of Cauchy surfaces 
$S_t$ with associated lapse-shift data $(N, X)$. 

\medskip 

  This brief summary of the covariant phase space method with boundary neglects however several significant issues. 
As discussed in \S 2, there are seemingly natural situations, e.g.~Dirichlet boundary data where $\cB = \Met(\cC)$, for 
which the identification \eqref{ES} fails badly. For generic boundary data $b \in \cB$, one may have $\w \bP_b = \emptyset$. 
Moreover, the structure of $\w \bP_b$ may be highly unstable with respect to variations of $b \in \cB$. 
For the purely formal considerations outlined above to accurately reflect the behavior of the space of solutions $\w \bP$ or $\w \bP_b$ 
(or their symplectic quotients), these spaces should be smooth infinite dimensional manifolds (or with small, understood singular 
sets) with the submanifolds $\w \bP_b$ varying smoothly with $b$ (at least generically). For this, one needs at least the {\it a priori} 
existence of charts for these spaces effectively parametrizing them locally. 

   It appears that the only method to do this is to prove that the IBVP associated to $\w \bP_b$, $b \in \cB$, is well-posed. In this case, 
one has natural parametrizations 
\be \label{mod4}
\hat \bP = \w \bP/\mathrm{Diff}_0(M) \simeq \cI \times_{c} \cB,
\ee
as in \eqref{mod2}. As discussed in \S 2, this does not hold for Dirichlet or Neumann boundary data. Further, based on the analogous 
result for closed or asymptotically flat Cauchy surfaces $S$, cf.~\cite{FM}, one would expect that
$$\hat \bP \simeq \bP,$$
at least generically. 

\begin{remark} 
{\rm An alternate path to the definition of Hamiltonians is the Hamilton-Jacobi method, used for instance by Brown-York \cite{BY}, 
cf.~also \cite{Sz} for a survey.  Here one also considers the space $\w \bP$ of (maximal) globally hyperbolic vacuum solutions on $M$ 
and a corresponding space $\cB$ of boundary data (Dirichlet data in the case of \cite{BY}). Analogous to the Hamilton-Jacobi method in 
classical mechanics, a Hamiltonian is defined as the variation of the on-shell action \eqref{I} with respect to a choice of time translation $\p_t$ 
of the boundary data $b \in \cB$, expressed in terms of a lapse-shift $(N, X)$ decomposition of $\cC$. For this approach to be effective, 
one needs to know that the action $\cI: \w \bP_b \to \bR$ is a smooth function of $b$ and so in particular the space 
$\w \bP_b$ varies smoothly with $b$. Again, as discussed in \S 2, this cannot be expected to be true for Dirichlet 
(or Neumann) boundary data. 

}
\end{remark} 

\begin{remark}
{\rm A shortcoming of the Hamiltonian approach above is that Hamiltonians associated to infinitesimal generators $\xi \in T\mathrm{Diff}(\cC)$ 
only exist for $\xi$ preserving the boundary condition, i.e.~$\xi \in T\mathrm{Diff}_b(\cC)$. For generic boundary data, no such $\xi$ 
exist and so the theory applies only in very special circumstances. 

 One would thus like to develop a well-defined Hamiltonian for arbitrary $\xi \in T\mathrm{Diff}(\cC)$. Note first that it 
follows easily from the constraint equations that $\d I (L_{\xi}g) = 0$ for any vector field $\xi$ of 
compact support on $\cC$ (for any diffeomorphism invariant boundary action). More generally, for $h = L_{\xi}g$, the constraint 
equations show that the boundary condition \eqref{bc3} may be rewritten in the form 
\be \label{di}
\div (\tau(\xi)) + \div(\ell \xi) + \div h(\nu)^{T} - dC(h) = 0 \ \ {\rm on} \ \ \cC.
\ee
Hence for infinitesimal diffeomorphisms, the boundary condition \eqref{bc} or \eqref{bc3} is only a condition at the 
future and past boundaries $\Si^{\pm}$. This suggests that ideally one should seek boundary conditions which are 
preserved under the action of $\mathrm{Diff}(\cC)$. 

  Nevertheless, \eqref{di} does not imply one can construct Hamiltonians with, for example, diffeomorphism-invariant Dirichlet boundary 
data, i.e.~boundary data in the quotient space $\Met(\cC)/\mathrm{Diff}(\cC)$. For such a larger space of admissible boundary data, 
Hamiltonians will no longer be conserved and it is necessary to work in the context of time-dependent Hamiltonians. This has been done 
using extended phase space techniques in work of Booth \cite{Bo} and Booth-Fairhurst \cite{BF}. 
 
}
\end{remark}

The covariant phase space method applied to the $\mathrm{Diff}(\cC)$-invariant boundary data in \cite{AA} discussed in \S 2 is rather 
different from the discussion above. To explain the situation, as before let $\w \bP$ be the space of all (maximal, globally hyperbolic) 
vacuum Einstein metrics on $M$. Given a Cauchy surface $S \subset M$, and a choice of boundary vector field $\Theta_{\cC}$ on 
$\cC_0$ as in \eqref{bd}, as discussed in \S 2 there is a unique wave map $\f_g: (M, g) \to (M_0, g_R)$ associated to $g$. As in 
\eqref{equi1}, the assignment $g \to \f_g$ is equivariant with respect to the action of $\mathrm{Diff}(M)$ on $\w \bP$. Let 
$\psi_g = \f_g^{-1}$ and let $\chi(\cC_0)$ be the space of vector fields on the boundary $\cC_0$. Define then a map 
\be \label{gt}
\begin{array}{c}
\w \bP \times \chi(\cC_0) \to \bar \bP,\\
(g, \Theta_{\cC}) \to \bar g = \psi_g^*g,
\end{array}
\ee
where $\bar \bP$ is the space of vacuum metrics on $M_0$ in the standard $(t, x^i)$ Minkowski coordinates. 
The group $\mathrm{Diff}(M)$ acts on the first factor on the left and acts trivially on the second and the map \eqref{gt} 
descends to the quotient to give a bijection 
$$(\w \bP \times \chi(\cC_0))/\mathrm{Diff}(M) \to \bar \bP.$$
Variations $\bar h$ of $\bar g$ are of the form 
\be \label{wh}
\bar h = \psi_g^*h - L_{X_{h}}\bar g - L_{X_{\Theta'}}\bar g \in T\bar \bP,
\ee
where $h$ is an on-shell variation of $g$ and $X_h$, $X_{\Theta'}$ are the variations of the wave map $\f_g$ induced by the 
variation $h$ of $g$ (with $\Theta_{\cC}$ fixed) and variations $\Theta'$ of $\Theta_{\cC}$ with $g$ fixed. Regarding the latter, 
fix a vacuum solution $g$ and consider variations $\Theta_{\cC}(t)$ of the Sommerfeld boundary data $\Theta_{\cC}$ for $\f_g$, 
satisfying the compatibility conditions at the corner $\Si$, holding the initial conditions fixed. One has unique wave map 
solutions $\f_g(t)$ of \eqref{wm}-\eqref{bd} realizing the boundary conditions $\Theta_{\cC}(t)$. By uniqueness, the pair 
$(g, \f_g(t))$ are solutions of the coupled system \eqref{couple}-\eqref{bcg} satisfying the initial data and boundary data 
$([\g^t], H^*, \Theta_{\cC}(t))$. This gives a curve of gauges, i.e.~3+1 decomposition or space-time foliation of the fixed metric 
$g$. Then  
\be \label{fvar}
X_{\Theta'} = \frac{d}{dt}\f_g(t)|_{t=0},
\ee
where $\Theta' = \frac{d}{dt}\Theta_{\cC}(t)$. 
 
  One may now carry out the covariant phase space method on the space $\bar \bP$. The definition of the corresponding pre-symplectic 
current $\bar \t$ is just as in \eqref{theta}, acting on $\bar h \in T\bar \bP$. Similarly $\bar \o$ is defined as in \eqref{omega}, leading to 
the (pre)-symplectic form 
\be \label{symaa}
\bar \O(\bar h_1, \bar h_2) = \int_{S_0} \<\pi'_{\bar h_1}, \bar h_2\> - \<\pi'_{\bar h_2}, \bar h_1\>,
\ee
where $S_0 = \{t = 0\}$. Although $d \bar \o = 0$, $\bar \o$ is non-zero on $\cC$ in general, without any boundary condition, so that 
$\bar \O$ is not conserved, i.e.~independent of the Cauchy slice $S_0$. 

   A boundary action and associated boundary condition is discussed below in \S 4, cf.~Proposition 4.3. We also introduce there an 
expanded phase space allowing for the existence of non-trivial Hamiltonians; see Propositions 4.4 and 4.5 below.

\section{Boundary actions in GR} 
\setcounter{equation}{0}

  In this section, we consider several boundary actions $\int_{\cC}\ell$ for the Einstein-Hilbert action $I_{EH}$ and discuss  
the resulting Hamiltonians generated by space-time diffeomorphisms. 

  The physics literature has long (only) considered the Gibbons-Hawking-York boundary term $\frac{1}{8\pi}\int_{\cC}Hdv_{\g}$ 
and the corresponding action  
$$I_{GHY} = I_{EH} + 2\int_{\cC}Hdv_{\g},$$
(where again we set $16\pi G = 1$). The variation of the boundary Lagrangian is $(H'_h + \frac{1}{2}\<H\g, h^T\>)dv_{\g}$, 
where $h^T$ is the variation of the boundary metric $\g$. Inserting this in \eqref{bc3} gives the well-known 
formula 
\be \label{tau}
\d I_{GHY} = -2\int_{\cC}\<\tau, h^T\>dv_{\g},
\ee
on $\w \bP$ where $\tau = A - H\g$ is the momentum conjugate to the boundary metric $\g$. The form $\tau$ is also the 
Brown-York stress-energy tensor of the boundary $\cC$, i.e.~the variation of the action with respect to variation of the Dirichlet boundary 
data, cf.~\cite{BY}, \cite{HW}.

  This gives a well-defined variational problem, i.e.~\eqref{bc} holds, for Dirichlet boundary conditions $h^T = 0$, so that 
$\cB = Met(\cC)$. However, as discussed in \S 2, Dirichlet boundary conditions are not well-posed for the IBVP; for 
generic boundary data $\g$, there is not any vacuum solution realizing $\g$, i.e.~$\w \bP_{\g} = \emptyset$. 

  Nevertheless, one may formally compute the Hamiltonian of a vector field $\xi$ tangent to $\cC$ preserving the boundary 
condition, i.e.~$\xi$ a Killing field on $(\cC, \g)$. Simple computation from \eqref{Ham2} gives 
\be \label{byh}
\cH_{\xi}^{Dir} =  -2\int_{\Si}\tau(\xi, T)dv_{\g_{\Si}}.
\ee
This is the Dirichlet or Brown-York ``bare" Hamiltonian; normalizations of $\cH_{\xi}^{Dir}$ corresponding to a choice of zero-point 
energy are discussed further in \S 5. One may of course always define \eqref{byh} formally or by fiat, but the Hamiltonian interpretation 
is only well-defined in the special circumstances where the boundary $(\cC, \g)$ admits a Killing field, cf.~Remark 3.2. 
Recall also from \S 3 that in this case the Hamiltonian is independent of the choice of slice $\Si$.

  Consider next the boundary conditions $([\g], H)$ in \eqref{gH} which are conjectured to be well-posed for the IBVP, 
cf.~also \cite{A2}.  

\begin{proposition} 
The boundary conditions $([\g], H)$ have a well-defined variational formulation. In fact, the action 
\be \label{ch}
I_{CH} = \int_M R dv_g + \frac{2}{3}\int_{\cC}H dv_{\g},
\ee
has variational derivative given by 
\be \label{vch}
\d I_{CH}(h) = -\int_M\<E, h\>dv_g - \int_{\cC}(\<\tau_0, h^T_0\> + \frac{4}{3}H'_h)dv_{\g},
\ee
where the subscript $0$ denotes the trace-free part with respect to $(\cC, \g)$. 

  If $\xi$ is a vector field on $\cC$ which preserves the boundary conditions, i.e.~$\xi$ is a conformal Killing field preserving $H$, then 
\be \label{chh}
\cH_{\xi}^{CH} = -2\int_{\Si}\tau_0(\xi, T)dv_{\g_{\Si}}.
\ee
\end{proposition}

\begin{proof} The first statement follows from \eqref{tau} and the simple computation $-\frac{4}{3}(\int Hdv_{\g})_h' = 
-\frac{4}{3}\int (H'_h + \frac{1}{2}\tr h) dv_{\g}$.  The second statement also follows easily from \eqref{Ham2}. 

\end{proof}

\begin{remark}
{\rm 
The on-shell actions $I_{GHY}$ and $I_{CH}$ are essentially the same; one has $I_{CH} = {\tfrac{1}{3}}I_{GHY}$ on $\w \bP$,
so that on-shell
\be \label{ot}
\d I_{CH} = {\tfrac{1}{3}}\d I_{GHY}.
\ee
However, the expressions \eqref{byh} and \eqref{chh} for the Hamiltonians associated to $\xi$ are not related in this simple way. 
The reason for this is the different boundary conditions; Dirichlet for $I_{GHY}$ and conformal class-mean curvature for $I_{CH}$. 
Distinct boundary conditions give different spaces $\w \bP_b$ and so different symplectic forms and hence different Hamiltonians. 

  Suppose $\xi$ is a Killing field on $(\cC, \g)$ which also preserves the mean curvature $\xi(H) = 0$, so that $h = L_{\xi}g$ preserves 
both both boundary conditions, i.e.~$h \in T(\w \bP_{\g})\cap T(\w \bP_{([\g],H)})$. By \eqref{Ham2}, the Hamiltonians $\cH_{\xi}^{Dir}$ and 
$\cH_{\xi}^{CH}$ have the same relation \eqref{ot} as the actions, 
\be \label{13}
\cH_{\xi}^{CH} = {\tfrac{1}{3}}\cH_{\xi}^{Dir},
\ee
if and only if the Noether charge vanishes, 
\be \label{q0}
Q_{\xi} = \int_{\Si}A(\xi, T)dv_{\g_{\Si}} = 0.
\ee
We note that it can be shown that if $\xi$ extends to a (stationary) Killing field of the bulk solution $(M, g)$, then \eqref{q0} does hold. 
However, it appears unlikely that \eqref{q0} holds in general. 
}
\end{remark}

  Next we point out that the $\mathrm{Diff}(\cC)$-invariant boundary data $([\g^t], H^*)$ from \cite{AA} may be given a 
variational formulation for a natural choice of $H^*$, at least when the Sommerfeld boundary condition \eqref{bd} for the 
gauge $\Theta_{\cC}$ is modified to the Dirichlet boundary condition 
\be\label{bm}
(\f_g)_*(T_g^c)= \hat \Theta_{\cC}.
\ee
Here $T_g^c$ is the future-pointing time-like unit normal to $\f_g^{-1}(\Si_t)\subset (\cC,g)$ and $\hat \Theta_{\cC}$ is a 
given vector field tangent to $\cC_0$. 

Recall from \S 2 that we work with vacuum solutions in preferred gauge $\f_g$ and let $\bar g = \psi_g^*g$, $\psi_g = \f_g^{-1}$, 
defined on the standard cylinder $M_0$. Let $\tr_{\Sigma}A$ be the trace of $A$ along the level set $\Si_t=\{t=constant\}\subset\cC_0$ 
of the boundary $\cC_0\subset (M_0,\bar g)$. 

\begin{proposition} The action 
\be\label{i}
I_{AA}(\bar g)=\int_{M_0}R_{\bar g} dv_{\bar g}+\int_{\cC} \tr_{\Si}A dv_{\bar \g},
\ee
gives a well-defined variational problem for the boundary data $([\g^t], H^*, \hat \Theta_{\cC})$, where 
\be \label{hst1}
H^* = 2\tr_{\cC}A - \tr_{\Si}A.
\ee
Thus, critical points of the action on the field space $\cF_b = \Met_b(M_0)$ with fixed boundary data $b = ([\g^t], H^*, \hat \Theta_{\cC})$ 
are the vacuum Einstein metrics with given boundary data $b$. 
\end{proposition} 

\begin{proof} Recall from \eqref{wh} that general variations $\bar h$ of $\bar g$ are of the form 
$\bar h=\psi^*_{g}h-L_{X_h}\bar g$. The variation of the bulk Lagrangian is 
\begin{equation*}
\begin{split}
&\d (R_{\bar g}dv_{\bar g})
=\<-E_{\bar g}, \psi^*h\>dv_{\bar g}+\<E_{\bar g}, L_X\bar g\>dv_{\bar g}+{\rm div}[{\rm div}\bar h- (d\tr\bar h)]dv_{\bar g}\\
&=\<-E_{\bar g}, \psi^*h\>dv_{\bar g}+{\rm div}[2E_{\bar g}(X)+{\rm div}\bar h-(d\tr\bar h)]dv_{\bar g},
\end{split}
\end{equation*}
where we have applied the Bianchi identity in the second equality. Thus the equations of motion generated by $I_{AA}$ are the 
vacuum Einstein equations. Using \eqref{h} or \eqref{bc3}, the on-shell variation of the action $I_{AA}$ with respect to deformation 
$\bar h$ is then easily computed to be 
\be\label{vi}
\d I_{AA}(\bar h) = \int_{\cC}(-\<A,\bar h\>-2(\tr_{\cC}A)'_{\bar h}+(\tr_{\Si}A)'_{\bar h}+\tfrac{1}{2}\tr_{\Si}A\tr_{\cC}\bar h)dv_{\bar \g}.
\ee
Let $\bar h_{\Si}$ be the restriction of $\bar h$ to $\Si_t$ and let $(\bar h_{\Si})_0$ denote its trace-free part. The first and fourth terms in 
\eqref{vi} may then be combined and rewritten as 
$-\<A, \bar h\> + \tfrac{1}{2}\tr_{\Si}A\<\bar \g, \bar h\> = -\<A, (\bar h_{\Si})_0\> - \<q(T), \bar h(T)\>$, where $q(T) = 
2A(T)_{\Si} + (-\tfrac{1}{2}\tr_{\Si}A+ A(T,T))T$. Also, setting $H^* = 2\tr_{\cC}A - \tr_{\Si}A$, i.e.~setting $\a = 2$, $\b = -1$, $\g = 0$ 
in \eqref{Hstar},  \eqref{vi} may be rewritten in the form 
\be\label{vi2}
\d I_{AA}(\bar h) = -\int_{\cC}\<A, (\bar h_{\Si})_0\> + (H^*)'_{\bar h} + \<q(T), \bar h(T)\>)dv_{\bar \g}
\ee
The vanishing of the terms $(\bar h_{\Si})_0$, $(H^*)'_{\bar h}$ and $\bar h(T)$ is equivalent to fixing the boundary data 
$([\g^t], H^*, \hat \Theta_{\cC})$. 

\end{proof}

  We note that for fixed $g$, the Dirichlet boundary condition \eqref{bm} still gives a well-posed IBVP for the wave map $\f_g$. It remains 
open however if the choice of $\hat \Theta_{\cC}$ in place of $\Theta_{\cC}$ leads to a well-posed IBVP for the coupled system 
$(g, \f_g)$ as in \eqref{couple}-\eqref{bcg}. We hope to investigate this elsewhere. 

  In this setting, the phase space $\bar \bP_b$ may be defined as in \eqref{gt} with associated pre-symplectic form $\bar \O$ in 
\eqref{symaa}. Unfortunately, this does not give rise to Hamiltonians associated with generators $\xi \in T\mathrm{Diff}(\cC)$. 
This is partly due to the fact that the boundary conditions are $\mathrm{Diff}(\cC)$ invariant but mainly due to the dependence of the 
preferred gauge $\f_g$ on the choice of the Cauchy surface $S$, as in \eqref{id}. 

\medskip 

 In order to obtain non-trivial Hamiltonians, we consider the same action on an enlarged space consisting of the pairs $(g,F)$ where $F$ 
denotes a general diffeomorphism $F:M\to M_0$ which maps $\cC$ to $\cC_0$. In the following let $\bar g =(F^{-1})^*g$ 
defined on $M_0$.  Choose fixed boundary data $b = \big( [\g^t], H^*,\hat\Theta_\cC, G\big)$, where $\big( [\g^t], H^*,\hat\Theta_\cC\big)$ 
are as above and $G: \cC \to \cC_0$ is an arbitrary but fixed diffeomorphism. Let ${\bf C}_b$ be the configuration space 
\be\label{Cgf}
 {\bf C}_b=\{(g,F):~\bar g=(F^{-1})^*g\mbox{ has boundary data } b = \big( [\gamma^t], H^*,\hat\Theta_\cC\big) \ {\rm with} \  F|_{\cC} = G\}.
\ee

\begin{proposition} The action 
\be\label{bari}
\bar I_{AA}(\bar g)=\int_{M_0}R_{\bar g} dv_{\bar g}+\int_{\cC_0} \tr_{\Si}A_{\bar g} dv_{\bar \g},
\ee
gives a well-defined variational problem on ${\bf C}_b$, for any boundary data $b$ with  
\be \label{hst2}
H^* = 2\tr_{\cC}A - \tr_{\Si}A.
\ee 
Thus, critical points of the action on ${\bf C}_b$ are the vacuum Einstein metrics with given boundary data $b$. 

  A vector field $\xi$ on $\cC$ which preserves the boundary conditions $([\g^t], H^*, \hat \Theta_{\cC})$ has an associated Hamiltonian  
\be \label{Hi2}
\cH_{\xi}^{AA}=\int_{\Si} -2A_{\bar g}(\bar \xi,T_{\bar g})+\<\bar\xi,T_{\bar g}\>\tr_{\Si} A_{\bar g}dv_{\bar\g_{\Si}},
\ee
where $\bar \xi = G_*(\xi)$. 
\end{proposition} 

\begin{proof} The proof of the first statement is the same as the proof of Proposition 4.3, with deformation $h$ of $g$ 
replaced by the deformation $(h, X)$ of $(g,F)\in{\bf C}_b$ and with corresponding deformation $\bar h$ of $\bar g$ on $M_0$ 
replaced by the deformation $\bar h=(F^{-1})^*h-L_X\bar g$ of $\bar g$. The same computations as before show that the equations 
of motion generated by $\bar I_{AA}$ are again the vacuum Einstein equations and that $\bar I_{AA}$ gives a well-defined 
variational problem on ${\bf C}_b$ for any $b$ satisfying \eqref{hst2}. 

  The resulting pre-phase space $\bar \bP_b$ is given by 
\be \label{barph}
\bar \bP_b =\{(g,F)\in{\bf C}_b:\mathrm{Ric}_g=0\}.
\ee
The pre-symplectic potential  $\theta$ is given by the same formula as in \eqref{theta} applied to $\bar h$:
$\bar \theta (\bar h)=\star_{\bar g}[\div \bar h-(d\tr\bar h)]$. This leads as before to the pre-symplectic form as in \eqref{symp2}:
\be\label{sympi2}
\bar \O\big( (h_1,X_1), (h_2,X_2)\big) = -\int_S \<\bar \pi'_{\bar h_1}, \bar h_2\> - \<\bar \pi'_{\bar h_2}, \bar h_1\>,
\ee
where $\bar h_i=(F^{-1})^*(h_i)-L_{X_i}\bar g$.

Now let $\xi$ be a vector field on $(M,g)$ such that the deformation $(L_\xi g,X)\in T\bar \bP_b$ for some deformation $X$ of $F$, so that 
$X = 0$ on $\cC$. Then $\xi$ induces a deformation of the pullback metric $\bar g$ on $M_0$ by $\bar h=L_{\bar \xi}\bar g$, where 
$\bar \xi =  F_*(\xi)-X$ on $M_0$. The Hamiltonian generated by $\xi$ is given by the same formula as in \eqref{Ham2} applied to the 
field $\bar\xi$:
$$\cH_{\xi}^{AA}=\int_\Sigma -2A_{\bar g}\big(\bar\xi,T_{\bar g}\big)+\<\bar\xi,T_{\bar g}\>\tr_\Sigma A_{\bar g}dv_{\bar\g_{\Si}}.$$
This gives \eqref{Hi2} since $F = G$ and $X = 0$ on $\cC$. 

Therefore we obtain a nontrivial Hamiltonian which is well-defined when $(L_\xi g,X)\in T\bar \bP_b$, for some $X$. Since $X=0$ on $\cC$, 
this is equivalent to the existence of a vector field $\xi$ on $(M_0,\bar g)$ such that 
\be \label{aabc}
\begin{split}
[L_{\xi} \g^t]_0 = 0,\ \ \xi(H^*) = 0,\ \  L_{\xi}\hat \Theta_\cC = 0
\quad{\rm on}~\cC.
\end{split}
\ee
A simple example is when $\xi = \p_t$ and the data $([\g^t], H^*, \hat \Theta_{\cC})$ are independent of $t$, arising for instance as 
boundary data of a stationary vacuum solution. 

\end{proof} 

  Note that the action $\bar I_{AA}$ imposes no equation of motion on the auxiliary field $F$. Also $F$ does not contribute to the 
Hamiltonian beyond the fixed boundary condition that $F = G$ at $\cC$. Next we point out that $F$ disappears when passing to the 
associated phase space and the resulting phase space agrees with the (standard) phase space defined in \eqref{qg}. 

\begin{proposition}
The phase space 
\be \label{barp}
\bP_b = \bar \bP_b/\cG,
\ee
associated to \eqref{barph} agrees with the vacuum phase space in \eqref{qg}.  
\end{proposition}

\begin{proof}  First, observe that for any deformation $X$ of $F$ with $X = 0$ on $\cC$, $(0,X)\in{\rm Ker}\,\w\O$, as in \eqref{og}. 
Thus one may first reduce $\bar \bP$ to the quotient space $\bP_1=\bar \bP/\cG_1$ where $\cG_1$ is generated by vector fields 
$(0,X)\in T\bar \bP_b$. If $(g,F)\in \bar \bP_b$ then $(g, F')\in \bar \bP_b$ for arbitrary diffeomorphisms 
$F':(M,\p M)\to(M_0,\p M_0)$ such that $F'=G$ on $\cC$. All of these are in the same equivalence class in $\bP_1$.\footnote{One should 
restrict here to diffeomorphisms isotopic to the identity but we will forgo such distinctions to simplify the discussion.} A natural 
choice of representative of an equivalence class is given by the pair $(g,F)$ with $F$ solving
\begin{equation*}
\begin{split}
\Box_g F=0\ \ {\rm in}~M,\ \ F=E_0,~F_*(T_g)=E_1\ \ {\rm on}~S,\ \ F=G\ \ {\rm on}~\cC,
\end{split}
\end{equation*}
for a fixed but arbitrary choice of initial data $(E_0, E_1)$; compare with \eqref{wm}-\eqref{bd}. This gives \eqref{barp}, but it 
worthwhile to carry the analysis somewhat further. 

  From the expression of $\bar \O$ in \eqref{sympi2}, it is clear that if the deformation $(h,X)$ preserves the initial data of 
$\bar g$ on $S$, i.e.~$(\bar g_S)'_{\bar h}=0~(K_{\bar g})'_{\bar h}=0$, then $(h,X)\in{\rm Ker}\,\bar \O$. Therefore one may further 
reduce the space $\bar \bP_1$ to 
$$\bar \bP_2 = \bar \bP_1/\cG_2,$$
where $\cG_2$ is generated by such deformations $(h,X)$. 

  Next, assume $[(g_t,F_t)]\in \bP_1$ is a smooth curve of equivalence classes in $\bP_1$ where the initial data of $g_t$ varies 
along a curve of diffeomorphisms: $g_t|_S=\psi_t^*(g_0|_S)$ and $K_{g_t}=\psi_t^*(K_{g_0})$ for $\psi_t\in{\rm Diff}_0(S)$ with 
$\psi = \rm{Id}$ on $\Si$. Consider a new curve $F'_t$ such that $F'_t$ solves 
\begin{equation*}
\begin{split}
\Box_{g_t}F'_t=0\ \ {\rm in}\ \  M,\ \ F'_t= F_0\circ\psi_t\ \ {\rm on }\ \ S,\ \ F'_t=G\ \ {\rm on}\ \ \cC.
\end{split}
\end{equation*}
Clearly $(g_t,F_t)\sim (g_t, F'_t)$ in $\bP_1$. In addition, corresponding to the family $(g_t,F'_t)$, the pullback metrics 
$\bar g_t=((F'_t)^{-1})^*g_t$ with $(F'_t)^{-1} = \psi_t$ have the same initial data, and hence the infinitesimal deformations along this curve 
belong to $\cG_2$. Thus in $\bar \bP_2$, $(g_1,F_1)\sim (g_2,F_2)$ if $g_1$ and $g_2$ have equivalent initial data, modulo 
diffeomorphisms equal to the identity on $\cC$.\footnote{As above with 
diffeomorphisms, we assume there exists a curve $g_t$ connecting $g_1$ and $g_2$.} In particular, $\bP_b$ is a further quotient 
(possibly trivial) of $\bar \bP_2$.  

\end{proof}

  As with \eqref{ES} one would like to establish a one-to-one correspondence between $\bP_b$ and $\cI/{\rm Diff}_b(S)$, for 
$b = \big( [\g^t], H^*,\hat\Theta_\cC)$ as above. This would require proving the well-posedness of the following IBVP:
\be\label{ibvp}
\begin{split}
&\mathrm{Ric}_g=0,~\Box_gF=0\ \ {\rm in}~M,\\
&g|_S = g_S,~K_g=K,~F=E_0,~F_*(T_g)=E_1\ \ {\rm on}~S,\\
&[\bar g^t]=[\g^t],~-2\tr_\cC A_{\bar g}+\tr_{\Sigma}A_{\bar g}=H^*,~T^c_{\bar g}=\hat\Theta_\cC,~F=G\ \ {\rm on}~\cC.
\end{split}
\ee
The system \eqref{ibvp} with Dirichlet boundary condition $\hat\Theta_\cC$ replaced by the Sommerfeld boundary condition 
$(T_{\bar g}+\nu_{\bar g})^T=\Theta_{\cC}$ as in \eqref{bd} on $\cC$ was proved to be well-posed in \cite{AA}. As discussed 
following Proposition 4.3, we hope to discuss the well-posedness of \eqref{ibvp} elsewhere.

\begin{remark}
{\rm Based on the geometric uniqueness result in \cite{AA}, if $(g_1,F_1),(g_2,F_2)$ are solutions to \eqref{ibvp} with the same initial data 
and boundary data, and with $\Theta_{\cC}$ in place of $\hat \Theta_{\cC}$, then they are equivalent in the sense that one is the 
pullback of the other by a diffeomorphism in $\mathrm{Diff}_0(M)$ (fixing $S$ and $\cC$); in particular $g_1$ and $g_2$ are isometric 
in $\mathrm{Diff}_0(M)$. This is consistent with the equivalence relationship in $\bar \bP_b$. 

  On the other hand, if $(g_1,F_1)$ and $(g_2,F_2)$ are solutions to \eqref{ibvp} (again with $\Theta_{\cC}$ in place of $\hat \Theta_{\cC}$) 
with all the same initial and boundary data except that $g_1=\psi^*(g_2)$, $K_{g_1}=\psi^*(K_{g_2})$ on $S$ for some diffeomorphism 
$\psi\neq {\rm Id}_S$ of $S$, it is unknown whether $g_1,g_2$ are equivalent, i.e.~isometric. This is the motivation of introducing the 
preferred wave map in \cite{AA}. However, such elements are equivalent in the phase space $\bar \bP_b$. 
}
\end{remark}

\section{Normalized Hamiltonians}
\setcounter{equation}{0}

As noted in the Introduction, several approaches to the definition of quasi-local energy-momentum or angular momentum of a space-like 
2-surface $\Si$ are based on a Hamiltonian or the closely related Hamilton-Jacobi approach. This was initiated by Brown-York \cite{BY}, 
with modifications developed by numerous authors including Epp, Kijowski, Liu-Yau, Wang-Yau and others; we refer to \cite{Sz} for a detailed 
survey and further references. For all of these definitions, the space $\cB$ of boundary data is the Dirichlet data space $\cB = \Met(\cC)$; the 
various definitions of quasi-local energies (or momenta) differ by a choice of normalization term, i.e.~a choice of zero-point energy (or momenta). 
 
  In general, the starting point of such a Hamiltonian approach is a space-like surface $\Si$ with data on $\Si$ obtained by fixing a choice of 
boundary data $b \in \cB$ restricted to $\Si$. One assumes then that $\Si$ with such data can be extended as a Cauchy slice in a time-like 
hypersurface $\cC = \p M$, where $(M, g)$ is a vacuum space-time. Moreover, one assumes the extension to $\cC$ has a vector 
field $\xi \in T\cC$ with flow preserving the boundary condition $b \in \cB$. For Dirichlet boundary data, this means $\xi$ is a 
Killing field of $(\cC, \g)$, $L_\xi \g = 0$. If $\xi$ is time-like, then the Hamiltonian $\cH_{\xi}$ serves as a notion of bare 
quasi-local energy.

  Recall from \eqref{byh} that the bare Brown-York quasi-local energy is given by 
\be \label{hby}
\cH^{Dir} = -2\int_{\Si}\tau(\xi, T)dv_{\g_{\Si}}, 
\ee
where $T$ is the unit time-like normal to $\Si \subset \cC$. Here the data on $\Si$ consist of the specification of a time-like 
3-metric $\g$ on $\cC$ (Dirichlet data) restricted to $\Si$ for which $\xi$ is a Killing field at $\Si$. The initial choice of Brown-York 
is simply to take 
\be \label{prod}
\g = -dt^2 + \g_{\Si},
\ee
with $\xi = T = \p_t$ representing time translation. A simple computation shows that 
\be \label{hby2}
\cH^{Dir} = -2\int_{\Si}H_{\Si}dv_{\g_{\Si}}, 
\ee
where $H_{\Si} = \tr_{\Si}A$. 

 This expression is of course non-vanishing for general surfaces $\Si \subset \bR^3$ and the prescription of Brown-York is to 
normalize the bare Hamiltonian by subtracting its Euclidean value:
\be \label{hby3}
\cH^{BY} = 2\int_{\Si}(H_{\Si}^0 - H_{\Si})dv_{\g_{\Si}},
\ee
where $H_{\Si}^0$ is the mean curvature of an isometric immersion $\iota (\Si, \g_{\Si}) \to \bR^3$. Thus, the subtraction term is 
the Brown-York Hamiltonian of $\cC_0 = \bR\times \iota(\Si) \subset \bR^{1,3}$ with respect to the standard time translation $t$ 
in Minkowski space. At least in the case of $\g_{\Si}$ with positive Gauss curvature, such an embedding exists by fundamental work 
of Nirenberg and Pogorelov, and is unique (up to congruence) by fundamental work of Cohn-Vossen.  Unfortunately, little is known 
about such existence and uniqueness results outside the case of positive Gauss curvature. 

  More generally, the Wang-Yau energy \cite{WY} considers families of isometric embeddings $(\Si, \g_{\Si}) \to \bR^{1,3}$ and 
chooses the reference cylinder $\cC_0$ to be spanned by arbitrary Minkowski time-translation fields transverse to $\iota(\Si)$. This may be 
viewed as a space-time generalization of the Brown-York prescription and requires developing an understanding of the possible 
isometric immersions $\Si \to \bR^{1,3}$, carried out in detail in \cite{WY}.   
  
\medskip   
  
   Returning to the general expression \eqref{hby}, since Dirichlet data are not well-posed, it is of course not well-understood when 
$(\Si, \g_{\Si})$ can be embedded in $(\cC, \g) \subset (M, g)$ with time-like Killing field $\xi$ on $\cC$, as in \eqref{prod} for instance. 
Thus, \eqref{hby} can only be considered as being formally defined, not well-defined. 

  Note also that even if \eqref{hby} is defined, i.e.~one has a vacuum space-time $(M, g)$ with boundary $(\cC, \g)$ extending $\Si$ and 
possessing a time-like Killing field $\xi$, the Euclidean (Brown-York) or Minkowski (Wang-Yau) subtraction term above may not be defined. 
To make this explicit, consider the space $\Met_0(\cC)$ of Lorentz metrics on $\cC$ with a time-like Killing field $\xi$. Such metrics are 
parametrized by metrics $\Met(\Si)$ on $\Si$, together with a choice of lapse and shift $(N, X)$ on $\Si$. This data is time-independent, 
i.e.~extended to be invariant under the flow of $\xi$, but the data $(N, X)$ may be arbitrarily chosen over $\Si$. However, in $\bR^{1,3}$ 
there is only a 4-dimensional space of time-translation Killing fields; thus most Killing fields on $(\cC, \g)$ cannot be renormalized to 
zero-point energy in this way. 
 
 \begin{remark}
{\rm When the IBVP is well-posed so that \eqref{ES} holds, a Hamiltonian $\cH_{\xi}: \bP_b \to \bR$ becomes a function on the space of 
initial data $\cH_{\xi}: \cI \to \bR$. The normalization terms above then correspond to solutions evaluated on flat initial data 
$(g_S, K)$. Recall that $(g_S, K)$ represent flat initial data if and only if the flat constraint equations hold; in schematic form 
\be \label{flat}
\begin{array}{c}
dK = 0,\\
R_{g_S} + K^2 = 0.
\end{array}
\ee
The traces of these equations give the vacuum Einstein constraint equations \eqref{scal}-\eqref{div}. If \eqref{flat} holds, then by the 
fundamental theorem of hypersurfaces in space-forms, $S$ has an isometric immersion into $\bR^{1,3}$ unique up to isometry of 
$\bR^{1,3}$, (at least when $S$ is simply connected). In the time-symmetric (Brown-York) case, $K = 0$ and so $g_S$ is a flat metric 
on $S \subset \bR^3$, or more precisely an immersion of $S$ into $\bR^3$. However, in general it is not well understood when a 
given geometry at the boundary $\Si$, i.e.~boundary data $b \in \cB$, has a fill-in by a flat metric on $S$, or when such fillings are unique. 

}
\end{remark}

 We consider instead a natural definition of the normalization term that is formally valid in general:
  
\begin{definition} 
Given a Hamiltonian $\cH_{\xi}: \bP_b \to \bR$ associated with a boundary condition $b \in \cB$, define the normalized Hamiltonian by 
$$\hat \cH_{\xi} = \cH_{\xi} - \inf_{\bP_b}\cH_{\xi}.$$
For this to be well-defined, one needs of course 
\be \label{ninf}
\inf_{\cI}\cH_{\xi} > -\infty.
\ee
\end{definition}
When the identification \eqref{ES} holds, so that $\bP_b \simeq \cI$, this gives the more concrete expression
\be \label{normHam}
\hat \cH_{\xi} = \cH_{\xi} - \inf_{\cI}\cH_{\xi}.
\ee
An advantage of this definition is that it immediately implies the positivity property:  
$$\hat \cH_{\xi} \geq 0 \ \ {\rm on} \ \ \cI.$$
Consider first the time-symmetric case where $K = 0$. If there exists a flat solution with time-symmetric initial data, i.e.~flat initial data 
$(g_S^0, 0) \in \cI$ filling in the boundary data $b \in \cB$, then it is natural to ask if the infimum $\inf_{\cI}\cH_{\xi}$ is realized by flat 
data: 
$$\inf_{\cI}\cH_{\xi} = \cH_{\xi}(g_S^0,0),$$
i.e.
$$\hat \cH_{\xi}(g_S^0, K^0) = 0.$$ 
It is also natural to consider the corresponding uniqueness or rigidity issue. For Dirichlet boundary data $\g_{\Si}$ with positive Gauss 
curvature and with $N = 1$ on $\Si$ (the Brown-York case as in \eqref{prod}-\eqref{hby2}), this is proved to be true by a basic result 
of Shi-Tam \cite{ST1}. A generalization of this result to the setting with non-zero $K$ (the space-time setting) has been proved in 
basic work of Wang-Yau \cite{WY}. 

\medskip 

  The term $\inf_{\cI}\cH_{\xi}$ bears a formal similarity with the definition of the Bartnik quasi-local energy \cite{Ba}, which is based on 
minimizing the ADM Hamiltonian for an asymptotically flat space-time in a region exterior to the Cauchy surface $S$; thus the interior 
minimization problem considered here is replaced in the Bartnik program by minimization over a complementary exterior region. 
Bartnik conjectures that such infima are realized by stationary vacuum solutions, i.e.~vacuum solutions for which the asymptotic 
time-like Killing field extends to a time-like Killing field over the full exterior region. It is of course natural to make the analgous 
conjecture in this (interior) setting. This will be discussed in more detail elsewhere. 

  In the case that the Killing field $\xi$ on $\cC$ is hypersurface orthogonal, i.e.~the shift $X$ of $\xi$ satisfies $X = 0$, it is proved in 
\cite{MST} that critical points of $\cH_{\xi}$ on $\cI$ are solutions of the static vacuum Einstein equations 
$$u\mathrm{Ric} = D^2 u, \ \  \D u = 0,$$
on $S$ with $u = N$ on $\Si$, partially confirming this conjecture. It is also proved in \cite{MM} that the bound \eqref{ninf} holds in 
this case.

\medskip 

  Next we consider analogs of the discussion above for other more well-behaved boundary conditions. 
  
   Consider first the $([\g], H)$ data \eqref{ch} which are conjectured to be well-posed. Let $\Si$ be any space-like 2-surface with metric 
$\g_{\Si}$ and let $\xi$ be a vector field on $\cC$ restricted $\Si$. For simplicity, consider here just extensions of this data to $\cC$ as in 
\eqref{prod}, so that $\xi = T = \p_t$. Assume also $H$ is independent of $t$. Then 
$$[\g] = [-dt^2 + g_{\Si}]$$ 
and $\xi$ is a conformal Killing field on $\cC$ preserving the mean curvature $\xi(H) = 0$, so $\xi$ preserves the boundary condition. 
Assuming the IBVP is well-posed for $([\g], H)$ (or instead just proceeding formally), this gives the Hamiltonian  
$$\cH_{\xi}^{CH} = -2\int_{\Si}\<A - \frac{1}{3}H\g, T\cdot T\>dv_{\Si} = -2\int_{\Si}(A(T,T) + \frac{1}{3}H)dv_{\Si}.$$
Since $H = -A(T,T) + H_{\Si}$ and $H$ is given boundary data, this may be re-written in the form  
\be \label{chham}
\cH_{\xi}^{CH} = -2\int_{\Si}(-\frac{2}{3}H + H_{\Si})dv_{\g_{\Si}}.
\ee
Both $H_{\Si}$ and the induced volume form $dv_{\g_{\Si}}$ are determined by the solution $g$, i.e.~initial data $(g_S, K)$.

 Here it appears more difficult to identify a suitable Euclidean subtraction term in general. One possibility is as follows; given $(\Si, \g_{\Si})$, 
as in the Brown-York case suppose there is a unique isometric embedding $\iota: (\Si, g_{\Si}) \to (\bR^3, g_{Eucl})$ into $\bR^3$. Let 
$H^0_{\Si}$ be the mean curvature of the image $\iota(\Si) \subset \bR^3$. Then {\it choose} the boundary condition $H = H^0 = H^0_{\Si}$. 
As above, there is a corresponding conformal embedding $(\cC, [\g]) \to (\bR^{1,3}, g_{Mink})$ and subtracting the Euclidean Hamilonian 
from \eqref{chham} gives 
$$\w \cH_{\xi}^{CH} = -2\int_{\Si}(-\frac{2}{3}(H - H^0) + (H_{\Si} - H^0_{\Si}))dv_{\g_{\Si}}.$$
Since by construction $H = H^0$, this gives 
$$\w \cH_{\xi}^{CH} = -2\int_{\Si}(H_{\Si} - H^0_{\Si})dv_{\g_{\Si}},$$
which is exactly the Brown-York energy again (without the $1/3$ term as in \eqref{13}). One advantage to this approach is that the 
boundary data $([\g], H)$ appear to be better behaved than Dirichlet boundary data. On the other hand, it is not clear if there is any 
suitable Euclidean subtraction term if $H$ is chosen arbitrarily. In this case, the definition \eqref{normHam} may be more suitable. 

 \medskip

  Finally consider the $([\g^t], H^*)$ data from \cite{AA} with $H^* = 2\tr_{\cC}A - \tr_{\Si}A$ as in Proposition 4.4. For simplicity, 
extend given data $([\g^t], H^*)$ on $\Si$ to be independent of $t$ on $\cC$ and set $\xi = \p_t = T  = \hat \Theta_{\cC}$, $G = Id$. 
Then the boundary condition \eqref{aabc} holds and the resulting bare Hamiltonian \eqref{Hi2} is given by 
$$\cH_\xi^{AA}=\int_{\Si} [-2A_{\bar g}(T,T) - \tr_{\Si}A_{\bar g}]dv_{\bar\g_{\Si}}.$$
As in the Brown-York case, consider here for simplicity a Euclidean subtraction term. Thus, given $([\g^t], H^*)$ on $\Si$ as above, let 
$\iota: (\Si, [\g_{\Si}]) \to \bR^3$ be a conformal embedding of $\Si$ into $\bR^3$ with prescribed mean curvature $H^*$. A general 
theory for such embeddings (or more precisely branched immersions) when $H^* > 0$ is developed in \cite{A3}. Such an embedding 
problem is elliptic and in general is much better behaved than the isometric embedding problem associated to Dirichlet boundary data. 
As above, one may extend $\iota$ trivially in the $t$-direction to an embedding of the cylinder $\cC_0$ into $\bR^{1,3}$ with time-like 
Killing field $\xi = \p_t$. This gives a flat vacuum solution with the same boundary data $([\g^t], H^*)$ and with $H^* = H_{\Si}^0$ (since 
$A(T,T) = 0$ on the flat cylinder). This gives a Euclidean normalized Hamiltonian
$$\w \cH_\xi^{AA}=\int_{\Si} [-2A_{\bar g}(T,T)  - \tr_{\Si}A_{\bar g}]dv_{\bar\g_{\Si}} + \int_{\Si}H^*dv_{\iota^*g_{Eucl}},$$
or equivalently 
$$\w \cH_\xi^{AA} =\int_{\Si} H^*(dv_{\bar \g_{\Si}} + dv_{\iota^*g_{Eucl}}) -2\int_{\Si}\tr_{\Si}A_{\bar g}dv_{\bar \g_{\Si}}.$$

\bibliographystyle{plain}

\end{document}